\documentclass{article}
\usepackage{booktabs}       
\usepackage{amsfonts}       
\usepackage{nicefrac}       
\usepackage{microtype}      
\usepackage{xcolor}         
\usepackage{fullpage}

\usepackage{fullpage,times}
\usepackage{times}

\usepackage{mathtools} 
\usepackage{booktabs} 
\usepackage{tikz} 
\usepackage{algpseudocode}
\usepackage{float} 

\usepackage{thm-restate}

\usepackage{subfigure}

\newcommand{\maxmove}[1][]{\ifthenelse{\equal{#1}{}}{$r$-move $k$-partitioning}{$r$-move ${#1}$-partitioning}}
\newcommand{\stcutatmost}{Min $r$-size $s$-$t$ cut}
\def\OPT{\textsc{OPT}}

\usepackage{amsmath,amsfonts,amssymb,multirow}
\usepackage{times}
\usepackage{graphicx}
\usepackage{floatpag}
\usepackage{algorithm}
\usepackage{bbold}
\usepackage{enumitem}
\usepackage{calc}
\usepackage{tikz}
\usepackage{hyperref}
\usetikzlibrary{decorations.markings}
\tikzstyle{node}=[circle, draw, inner sep=0pt, minimum size=4pt, fill = black]
\newcommand{\node}{\node[node]}
\usepackage{graphicx}
\usepackage{tabularx}
\makeatletter
\newcommand{\multiline}[1]{%
  \begin{tabularx}{\dimexpr\linewidth-\ALG@thistlm}[t]{@{}X@{}}
    #1
  \end{tabularx}
}
\makeatother
\usepackage{mathtools}

\makeatletter
\def\BState{\State\hskip-\ALG@thistlm}
\makeatother

\newcommand{\floor}[1]{\lfloor #1 \rfloor}

\newenvironment{proofof}[1]{{\bf Proof of #1.  }}{\hfill$\Box$}

\newtheorem{theorem}{Theorem}[section]

\newtheorem{corollary}{Corollary}[section]
\newenvironment{proof}{{\bf Proof.  }}{\hfill$\Box$}

\newtheorem{definition}{Definition}[section]
\newtheorem{lemma}{Lemma}[section]

\def \R {{\mathbb R}}

\newcommand{\ignore}[1]{}

\title{Graph Partitioning With Limited Moves}

\author{%
Majid Behbahani\thanks{Machine Learning Research, Morgan Stanley. Email: {\tt majid.behbahani@morganstanley.com.}} \qquad 
Mina Dalirrooyfard\thanks{Machine Learning Research, Morgan Stanley. Email: {\tt mina.dalirrooyfard@morganstanley.com.}} \qquad Elaheh Fata\thanks{Smith School of Business, Queen's University. Email: {\tt elaheh.fata@queensu.ca.}} \qquad Yuriy Nevmyvaka\thanks{Machine Learning Research, Morgan Stanley. Email: {\tt yuriy.nevmyvaka@morganstanley.com.}} 
}
\date{}

\begin{document}
\maketitle




\def\thefootnote{*}\footnotetext{Authors are ordered alphabetically.}

\begin{abstract}
In many real world networks, there already exists a (not necessarily optimal) $k$-partitioning of the network. Oftentimes, one aims to find a $k$-partitioning with a smaller cut value for such networks by moving only a few nodes across partitions. The number of nodes that can be moved across partitions is often a constraint forced by budgetary limitations. Motivated by such real-world applications, we introduce and study the $r$-move $k$-partitioning~problem, a natural variant of the Multiway cut problem. Given a graph, a set of $k$ terminals and an initial partitioning of the graph, the $r$-move $k$-partitioning~problem aims to find a $k$-partitioning with the minimum-weighted cut among all the $k$-partitionings that can be obtained by moving at most $r$ non-terminal nodes to partitions different from their initial ones. Our main result is a polynomial time $3(r+1)$ approximation algorithm for this problem. We further show that this problem is $W[1]$-hard, and give an FPTAS for when $r$ is a small constant. 

\end{abstract}

\section{Introduction}
\label{sec:intro}
Graph partitioning problems are among the most fundamental and widely used graph problems in the fields of artificial intelligence, theoretical computer science, operations research and operations management. A \textit{$k$-partitioning} of a graph $G$ refers to partitioning the graph's nodes into $k$-disjoint sets, clusters, or partitions. Oftentimes, the goal in partitioing problems is to find a $k$-partitioning with the least \textit{cut value} (also referred to as \textit{cut size} in the literature)
, which is the sum of the weights of the edges that are between different partitions.
Graph partitioning is widely used in machine learning, parallel computing, computer vision, VLSI design, political districting, epidemiology and many more areas \cite{vision,parallel1,learning,vlsi,king2012geo,epidemiology}.

Given a graph\footnote{We use the terms ``graphs" and ``networks", interchangeably.} $G$, an integer $k\ge 2$ and $k$ terminals, the well-known \textit{Multiway cut} problem asks to find a partitioning of the nodes of $G$ into $k$ clusters with the smallest cut value among all $k$-partitionings that have exactly one terminal in each partition \cite{dahlhaus1992complexity}. In this paper, we use the terms \emph{partitions} and \emph{clusters}, interchangeably. It is known that for $k>2$, the Multiway cut problem is APX-hard \cite{nphardmultiwaycut}. 

We introduce and study a natural variant of the Multiway cut problem called the \textit{\maxmove} problem. Suppose that we are given a graph $G$ with an initial $k$-partitioning, an integer $k\ge2$, $k$ terminals and a \textit{move parameter} $r\in \mathbb{N}$. The problem only considers $k$-partitionings for which at most $r$ non-terminal nodes have been moved from their initial partitions to new clusters. The \maxmove~problem asks to find a $k$-partitioning that minimizes the cut value over all the $k$-partitionings considered. 

Numerous examples of the \maxmove~problem can be seen in real-world networks. In fact, in many networks, the underlying graph already has a $k$-partitioning and this $k$-partitioning is not optimal, i.e., the cut value is not minimized. 
For instance, the edge weights in a network may change over time, resulting in the sub-optimality of the initial partitioning. However, in most applications, we cannot afford moving many nodes from their initial $k$-partitioning, and so the goal is to improve the cut value by moving only ``a few" nodes. 


For example, consider a company with many offices across the country. The company can be modeled as a graph, where the nodes are the employees and the edge weights indicate how beneficial it would be if these two employees are in the same office, for example it could be proportional to the amount they communicate, or any other metric. 
The company's goal is to find a $k$-partitioning with a low cut value, where $k$ is the number of available offices. Employees have initial assignments to these $k$ offices, which might not be optimal. However, due to budget constraints, companies are often unable to conduct large-scale reorganizations and can only move a limited number of employees between offices. Moreover, not all employees can or are willing to move. For each office, we can model these employees as a terminal node by ``contracting" all the employees of that office (partition) that cannot relocate. If the company's budget dictates that at most $r$ employees can be relocated, then we seek an \maxmove~where the initial partitioning is the original distribution of people across offices. Modifications of the office reorganization example are applicable to many other applications such as public housing relocation to improve the housing assignments and enhance the residents' community access, or disaster management (see \cite{dinitz2022fair}). 


Another application of the \maxmove~problem is its use as a sub-problem in a local search algorithm for the Multiway cut problem. Consider an initial feasible multiway cut $C$ from which we aim to achieve another feasible multiway cut $C'$ with a lower cost, allowing only $r$ moves for the transformation from $C$ to $C'$. This problem can be exactly formulated as the \maxmove~problem. For more information on local search approximation algorithms for the Multiway cut problem, the reader is referred to~\cite{bloch2023local}. Moreover, the \maxmove~problem can be explored in the context of learning-augmented algorithms. In this case, the initial multiway cut is obtained from a learning-based algorithm and the goal is to use this initial multiway cut to reach multiway cuts with lower costs, while moving at most $r$ nodes to partitions different from their initial ones. The question here is whether there exists an algorithm for the \maxmove~problem that can improve the initial multiway cut whose guarantees depend on the quality of the initial multiway cut. While this direction is interesting, it falls outside the focus of this paper.




There has been a long line of works providing approximation algorithms for the Multiway cut problem \cite{nphardmultiwaycut,LP_paper,angelidakis2017improved,vazirani,ugchardness}, most of which use a simple linear programming (LP) relaxation of the problem called the \textit{CKR LP}
, in honor of the authors who introduced it, C\u{a}linescu, Karloff and Rabani \cite{LP_paper}. The best approximation factor for the Multiway cut problem is $1.2965$ due to \cite{bestapprox}, who round the CKR LP. Interestingly, \cite{ugchardness} showed that assuming the unique games conjecture \cite{ugc}, if the integrality gap of the CKR LP is $\alpha$, then the Multiway cut problem cannot be approximated better than $\alpha$. \cite{berczi2020improving} showed that the integrality gap of the CKR LP is at least $1.20016$, hence the best approximation factor for the Multiway cut problem lies between $1.20016$ and $1.2965$. Closing this gap is an interesting open problem. 


Perhaps the closest known variant of the Multiway cut problem to the \maxmove~problem is the \textit{min $s$-$t$ cut with at most $r$ nodes}, also known as the \textit{\stcutatmost} problem \cite{zhang2016new}. This problem asks to find a minimum value $s$-$t$ cut, where there are at most $r$ nodes in the partition that contains terminal node $s$ (i.e., the $s$-side). This problem is a special case of the $(r-1)$-move $2$-partitioning problem, when in the initial partitioning all non-terminal nodes are in the partition that contains node $t$. By moving at most $r-1$ nodes to the $s$-side, the $s$-side will have size at most $r$.

Even though the \stcutatmost~problem is a variant of the $s$-$t$ cut problem, which is polynomially solvable, the \stcutatmost~problem is NP-hard \cite{chen2016size}. Thus, the \maxmove[2] and the \maxmove~problems are NP-hard. It is also known that the \stcutatmost~problem admits a \textit{Fixed-Parameter Tractable (FPT)} solution with parameter $r$ when the graph is unweighted \cite{lokshtanov2013clustering} 
 (see Section~\ref{sec:prelim} for a definition of FPT problems). Intuitively, one might wonder if the \maxmove[2] problem and the \stcutatmost~problem are equivalent. We show that the \maxmove[2] problem is in fact W[1]-hard; thus, \emph{not equivalent} to the \stcutatmost~problem. We resort to designing approximation algorithms for the \maxmove~problem for all $k\ge 2$. Hence, the main theme of this paper is around the following question:
\textit{How well can we approximate \maxmove?}

\subsection{Our Results}
%
We provide a comprehensive study of the \maxmove~problem. 
Recall that the \maxmove~problem is NP-hard (for any constant $k\ge 2$), hence, there exists no exact polynomial time algorithm for it, unless P = NP. Throughout the paper we think of $k$ as a constant.

First, note that when $r$ is also a constant, there exist a simple $O((nk)^{r})$ algorithm for the problem. This algorithm tries every combination of $r'$ nodes, for all $1\le r'\le r$, to be moved to any of the $k$ partitions. It is natural to ask whether we can devise an FPT algorithm with parameter $r$ for the \maxmove~problem. Our first result shows that the answer to this question is negative. 

\begin{restatable}{theorem}{wonehardness}\label{thm:w1hard}
The \maxmove~problem with parameter $r$ is W[1]-hard.
\end{restatable}

In Theorem \ref{thm:w1hard}, we use a polynomial time reduction from the Densest $r$-Subgraph problem, where the complete proof can be found in Section \ref{sec:hardness}. It is well-known that the Densest $r$-Subgraph problem with parameter $r$ is W[1]-hard \cite{densest}.

Next, we give a $(1+\epsilon)$-approximation algorithm for the \maxmove~problem with running time $f(r,\epsilon)O(n^2)$ for any $\epsilon>0$, where $f(\cdot,\cdot)$ is a function of $r$ and $\epsilon$. When $r$ is a constant, this gives us an FPTAS (with parameter $\epsilon$) with running time $O(n^2/\epsilon^r)$, faster than the $O(n^r)$ brute-force algorithm, in the expense of an approximation factor. 

\begin{restatable}{theorem}{fptas}\label{thm:fptas}
    Given any \maxmove~instance graph $G$ with $n$ nodes and any constant $\epsilon>0$, there exists a $(1+\epsilon)$-approximation algorithm for the \maxmove~problem on $G$ with running time $(\frac{2(k-1)(1+\epsilon)}{\epsilon})^rr!cn^2$, where $c$ is a universal constant independent of $r$ and $\epsilon$. 
\end{restatable}


The proof of Theorem \ref{thm:fptas} can be found in Section \ref{sec:fptas}. The running time of the algorithm in Theorem \ref{thm:fptas} grows fast as $r$ grows, which leads us to the following question: Can we design an approximation algorithm with a running time polynomial in both $n$ \emph{and} $r$? Our main result focuses on answering this question. We give a polynomial time approximation algorithm with approximation factor at most $3(r+1)$. This is done by extending the CKR linear program to an LP for the \maxmove~problem by adding a single linear constraint reflecting the move constraint of the problem. Our approximation algorithm is a simple rounding scheme for this LP, which makes our approach very practical. The main challenge is in proving the approximation guarantees of our rounding, which is given in Section \ref{sec:alg}. 

\begin{restatable}{theorem}{bestapprox}\label{thm:m2approx}
Given a positive integer $r$, there exists a randomized algorithm for the \maxmove~problem on any $n$-node $m$-edge graph $G$ that with approximation factor $\frac{2k}{k-1}(r+1)$ and running time $T_{LP}(n,m)+O(mk)$, where $T_{LP}(n,m)$ is the running time of solving a linear program with $O(n)$ variables and $O(m)$ constraints.
\end{restatable}

Firstly, note that we can de-randomize this algorithm in the expense of an additional factor of $n$ in the running time (see Section \ref{sec:derandom}). Secondly, note that unlike Theorem \ref{thm:fptas}, the running time of our algorithm in Theorem \ref{thm:m2approx} is independent of $r$. For $k=2$, we give a rounding scheme for our LP with approximation factor at most $r+1$, resulting in an $(r+1)$-approximation algorithm for the \maxmove[2]~problem (see Section \ref{subsec:2partitioning}).
We further show that the integrality gap of the LP used in Theorem \ref{thm:m2approx} is at least $r+1$, demonstrating that our rounding algorithm is tight within a constant factor (see Section \ref{sec:integrality-gap}).

Finally, as our approximation factor is dependent on $r$, we resort to bicriteria algorithms to make it constant.  
\begin{restatable}{theorem}{fourmove}\label{thm:fourM}
For any $1/2<\gamma<1$, there exists a polynomial time $(\frac{1}{1-\gamma},\frac{5}{2\gamma-1})$-approximation randomized algorithm for \maxmove~problem, where the first criterion is the number of nodes moved and the second criterion is the cut value. 
\end{restatable}


The proof of Theorem \ref{thm:fourM} can be found in Section \ref{sec:bich}. We provide two remarks 
to better understand the applicability of our results. Firstly, we show in Section \ref{sec:hardness} that \maxmove~without terminals is also W[1] hard and so even though for a lot of partitioning problems the ``without terminals" version is easier, this is not the case for \maxmove. Our second remark is stated below. We include simple complementary numerical evaluations of our algorithm in Section \ref{sec:exp}. 

\paragraph{Meaningful range for $r$:}
The range of $r$ in which Theorems \ref{thm:m2approx} and \ref{thm:fptas} 
provide meaningful results can be understood as follows. For the case of $k=2$, the Multiway cut problem is tractable and so it is possible to compute an optimal multiway cut $C$ for the input graph in polynomial time. Suppose that cut $C$ is obtained by moving $\hat{r}$ nodes from their initial partitions. Then, Theorems \ref{thm:m2approx} and \ref{thm:fptas} are meaningful only when $r<\hat{r}$. This is because if $r\ge \hat{r}$, then $C$ is an optimal output for the \maxmove[2] problem. On the other hand, when $k>2$, there is no polynomial time algorithm for the Multiway cut problem. Therefore, in this case, one can obtain a cut $C$ whose cut value is within a  constant approximation factor of an optimal multiway cut, using an approximation algorithm for the Multiway cut problem (for example the $\frac{3}{2}$-approximation algorithm of \cite{LP_paper}). Let $\hat{r}$ be the number of nodes that $C$ moves from their initial partitions. Once again, Theorems \ref{thm:m2approx} and \ref{thm:fptas} are useful only when $r<\hat{r}$.

\subsection{Related Works}
The Multiway cut problem without terminals is called the $k$-partitioning problem. While the $k$-partitioning problem is polynomially solvable for fixed $k$~\cite{goldschmidt1994polynomial}, the Multiway cut problem is NP-hard for $k\ge 3$~\cite{dahlhaus1992complexity}. 
There exist many ``budgeted" variants of the Multiway cut and $k$-partitioning problems in the literature, we name a few of these variants here.

\paragraph{The \stcutatmost~problem:} First, recall that the \stcutatmost~problem is NP-hard, even though the ``without terminals" version of this problem, that is, the Min $r$-size cut problem, is polynomially solvable \cite{armon2006multicriteria,watanabe1987edge}. The Min $r$-size cut problem asks to find a cut with the minimum value among all the cuts with one side having at most $r$ nodes. The result of \cite{armon2006multicriteria} is based on the result of \cite{karger1996new} which states that the number of cuts with value at most $\alpha$ times the value of the min cut is at most $n^{2\alpha}$. Such result does not exist for $s$-$t$ cuts and, in fact, the number of min $s$-$t$ cuts can be exponentially many\footnote{Consider a graph which consists of $n/2$ paths of length $2$ from terminal $s$ to terminal $t$. An $s$-$t$ cut has at least one edge from each path, and so the number of min $s$-$t$ cuts is $2^{n/2}$.}. The \stcutatmost~problem has been studied by other names such as the problem of ``cutting at most $r$ nodes by edge-cut with terminal" \cite{fomin2013parameterized,lokshtanov2013clustering}. Moreover, \cite{zhang2014unbalanced} gives an $O(\log(n))$-approximation algorithm for the \stcutatmost~problem using R\"{a}cke's tree decomposition method \cite{racke2008optimal}. See Table \ref{tab:2partition} for a comparison of the related variants of the \maxmove[2]~problem and their associated algorithms. 

Table \ref{tab:2partition} compares the variants of \maxmove~and \stcutatmost. Next we include omitted proofs.
\begin{table}[t]
\centering
\resizebox{\columnwidth}{!}{
\begin{tabular}{|l|l|l|l|}
\hline
                                      & FPT Apprx. Fctr& Apprx. Fctr & LB   \\ \hline
\stcutatmost                         & $1+\epsilon$*             & $\min(r+1,\log(n))$ \cite{zhang2016new,zhang2014unbalanced}       & $W{[}1{]}$-hard* \\ \hline
Min $r$-size cut                            & $1$ \cite{lokshtanov2013clustering}                       & $1$ \cite{lokshtanov2013clustering}         &               \\ \hline
\maxmove[2]                & $1+\epsilon$*             & $r+1$*         & $W{[}1{]}$-hard* \\ \hline
\maxmove[2] without terminals & $1+\epsilon$*             & $r+1$*         & $W{[}1{]}$-hard* \\ \hline
\end{tabular}
}
\caption{Known and new results on different variants of the $2$-partitioning problem. Asterisks are used to show the results provided by this paper. The second column indicates the best known approximation factors for fixed-parameterized approximation algorithms for the problems with parameter $r$. The third column is dedicated to the approximation factors of general approximation algorithms with running times polynomial in both $n$ and $r$, and the fourth column indicates lower bounds. Note that the NP-hardness of the \stcutatmost~problem was already known \cite{chen2016size} and this paper shows the $W[1]$-hardness of this problem in its weighted version.
}
\label{tab:2partition}
\end{table}

\paragraph{Other variants:}
Most of the variants of the Multiway $k$-cut problem that we are aware of are for $k=2$. Basically having a budget on the number of nodes in one side of the cut, and having terminals or no terminals in the graph makes up of most of these variants. 

The exact version of the \stcutatmost~problem is called the \textit{Min E$r$-Size $s$-$t$ cut} problem \cite{feige2003cutting}, which asks to find an $s$-$t$ cut with minimum cut value among all the cuts that have exactly $r$ nodes in the $s$-side. The \textit{$(r,n-r)$-cut} problem \cite{bonnet2015multi} is the above problem without any terminals. 
The $(r,n-r)$-cut problem is known to be W[1]-hard \cite{bonnet2015multi}, and as a result the Min E$r$-Size $s$-$t$ cut problem is also W[1]-hard. For both these problems, there is a randomized $O(r/\log(n))$-approximation algorithm due to \cite{feige2003cutting}. These problems have been studied when parameterized by cut value as well \cite{fomin2013parameterized}.


For $k>2$, the most relevant variant of the Multiway $k$-cut problem is the $k$-balanced partitioning problem, which is proved to be APX-hard~\cite{feldmann2015balanced}. Other problems that seem to be indirectly relevant are the $k$-route cut problem \cite{chuzhoy2015approximation} and the MinSBCC problem \cite{hayrapetyan2005unbalanced}.

~\cite{dinitz2022fair} study a variant of the min $s$-$t$ cut problem with fairness considerations. They introduce the Demographic Fair Cut problem, which is informally defined as follows: Given a terminal node $s$ and a labeling of nodes into various demographics, the goal is to find a minimum cut that disconnects at least a certain given fraction of each demographic from $s$. If all nodes of the graph belong to a single demographic, then this problem is equivalent to the \stcutatmost~problem. Thus, their $\log(n)$-approximation algorithm can also be applied to the \stcutatmost~problem, resulting in a $\min(r+1,\log(n))$-approximation algorithm for the \stcutatmost~problem. Whether a $\log(n)$-approximation algorithm exists for the \maxmove~problem (either for $k=2$ or $k>2$) is an interesting open question that is outside the scope of this paper.

\section{Preliminaries}\label{sec:prelim}

Let $G=(V,E)$ with weight function $c: E\rightarrow \R$, that is, an edge $(u,v)\in E$ has weight $c(u,v)$ or $c_{uv}$. We refer to an edge with endpoints $u$ and $v$ by $(u,v)$ or $uv$. We only consider undirected graphs throughout this paper. Let $E(G)$ and $V(G)$ denote the edge set and node set of a graph $G$.  A problem of size $n$ with parameter $r$ is said to be FPT if there is an algorithm that solves it in $f(r)O(n^c)$ time, where $f(\cdot)$ is an arbitrary function depending only on $r$ and $c$ is a constant. If any $W[1]$-complete problem is FTP, then FPT = $W[1]$ and every problem in $W[1]$ is FPT.

We formally define the \maxmove~problem using the definition of a $k$-cut (also known as $k$-partitioning). 
%
Given a weighted graph $G=(V,E)$ with weight function $c: E\rightarrow \R$, a \emph{$k$-cut} is a subset of edges $E'\subseteq E$ such that removing the set of edges $E'$ from the graph results in another graph $G'=(V,E\backslash E')$ that has $k$ connected components. 
We refer to these connected components that appear after removing edge set $E'$ as \textit{clusters} or \textit{partitions}. The weight of a $k$-cut $E'$ is the sum of the weight of the edges in $E'$ and it is called the \textit{cut value}. Any set of edges that introduce a cut in a graph, such as $E'$, can be referred to as a \textit{cut-set}.

\begin{definition}[The $r$-move $k$-partitioning problem]
Let $G=(V,E)$ be a weighted graph with a weight function $c: E\rightarrow \R$,  a set of terminals $S=\{s_1,\ldots,s_k\}\subseteq V$ and $|V|=n$. Suppose that $G$ has a given initial partitioning, where $\ell_v\in \{1,\ldots,k\}$ denotes the initial partition that node $v$ belongs to and for a terminal node $s_i$, we have $\ell_{s_i}=i$, for each $i\in \{1,\ldots,k\}$. The \textit{\maxmove}~problem asks to find a minimum-weighted $k$-cut $\ell^*:V\rightarrow \{1,\ldots,k\}$ such that $\ell_v^*\neq \ell_v$ for at most $r$ non-terminal nodes $v$.
\end{definition}


The parameter $r$ is referred to as the \textit{move parameter}. 
%
%
The linear program we use throughout this paper is represented in Table \ref{tab:maxmove-lp} and will be referred to as LP \ref{tab:maxmove-lp}. Suppose that $k$ is the number of partitions of interest. For each $v\in V$, let $X_v=(X_v^1,\ldots,X_v^k)$ be a vector of size $k$ of positive real decision variables in $[0,1]$. To understand the role of decision variable $X_v^i$ better, observe that if we were only considering integer solutions, then $X_v^i=1$ would represent node $v$ being in partition $i$, and constraint (C4) would ensure that each node $v$ is assigned to exactly one partition. We define the distance between pairs of vertices $u,v$ as $d(u,v)=\frac{1}{2}\sum_{i=1}^k |X_u^i-X_v^i|$. In case of integer solutions, if $u$ and $v$ are in the same partition, then $d(u,v)=0$ and if they are in different partitions, $d(u,v)=1$. This further motivates the objective function in LP \ref{tab:maxmove-lp}. To represent $|X_u^i-X_v^i|$ in the LP, we define variables $Y_{uv}^i$'s for each $(u,v)\in E$, and add constraints (C2) and (C3) to make sure $Y_{uv}^i = |X_u^i-X_v^i|$. Finally, we call constraint (C7) \emph{the move constraints}, which ensures that at most $r$ nodes are moved from their initial partitions, in case of integer solutions. Note that LP \ref{tab:maxmove-lp} without constraint (C7) is the CKR LP.

\begin{table*}[hht]
    \centering
    \begin{tabular}{lllr}
    \toprule
        \text{Minimize:} & $\sum_{(u,v)\in E}c_{uv}d(u,v)$ & &\\
        \text{Subject to:} & & &\\
        & $d(u,v)= \frac{1}{2}\sum_{i=1}^k Y^i_{uv}$  & $\forall (u,v)\in E$ & (C1)\\
        & $Y^i_{uv}\ge X^i_u-X^i_v$ & $\forall (u,v)\in E$, $\forall i\in\{1\cdots k\}$& (C2)\\
        & $Y^i_{uv}\ge X^i_v-X^i_u$ & $\forall (u,v)\in E$, $\forall i\in\{1\cdots k\}$& (C3)\\
        & $\sum_{i=1}^{k}X^i_u=1$ & $\forall u\in V$ & (C4)\\
        & $X_{s_i}=\mathbf{e}_i$ & $\forall s_i\in S$  & (C5)\\
        & $X_v^i \ge 0$ & $\forall v\in V$, $\forall i\in \{1,\ldots,k\}$ & (C6)\\
        &  $\sum_{v\in V}(1-X_{v}^{\ell_v}) \le r$ & & (C7)\\
        \bottomrule
    \end{tabular}
    \caption{$r$-move $k$-partitioning LP (LP~\ref{tab:maxmove-lp})}
    \label{tab:maxmove-lp}
\end{table*}



Suppose $X$ is a feasible, but not necessarily optimal, solution to 
LP~\ref{tab:maxmove-lp}. Notation $d_X(u,v)=\frac{1}{2}\sum_{i=1}^k |X_u^i-X_v^i|$ is used to better clarify that the distance function is calculated specifically for solution $X$. We use the notation $\bar{X}$ to denote the optimal solution of 
LP~\ref{tab:maxmove-lp}. 
%
We conclude this section with the following lemmas that will be later used in our proofs. 
\begin{lemma}\label{lem:size}
Consider a feasible solution ${X}$ to LP~\ref{tab:maxmove-lp}. Given a number $0\le\lambda<1$, if $L$ 
is the set of nodes $v$ for which ${X}_v^{\ell_v}<\lambda$, then $|L|< \frac{r}{1-\lambda}$.
\end{lemma}
\begin{proof}
For each $v\in L$, we have $1-{X}_v^{\ell_v}>1-\lambda$. Since ${X}$ is a feasible solution of LP~\ref{tab:maxmove-lp}, $r\ge\sum_{v\in L}(1-{X}_v^{\ell_v})$. So $r>(1-\lambda)|L|$, and $|L|< \frac{r}{1-\lambda}$.
\end{proof}


\begin{lemma}\label{lem:helper2}
Let $X$ be any feasible solution to LP~\ref{tab:maxmove-lp}. For any $i\in\{1,\ldots,k\}$ and $u,v\in V$ we have $d_X(u,v)\ge |X_u^i-X_v^i|$.
\end{lemma}
\begin{proof}
By the triangle inequality we have $\sum_{j\neq i}|X_u^j-X_v^j|\ge |\sum_{j\neq i} X_u^j-\sum_{j\neq i} X_v^j| =|(1-X_u^i)-(1-X_v^i)|$, where the last equality used $\sum_j X_u^j=\sum_j X_v^j=1$. Therefore, $d_X(u,v)=\frac{1}{2}(|X_u^i-X_v^i|+\sum_{j\neq i}|X_u^j-X_v^j|)\ge |X_u^i-X_v^i|$.
\end{proof}

\section{FPTAS
} \label{sec:fptas}
Here, we first provide high level intuition on the algorithm of Theorem \ref{thm:fptas} and prove it formally in section~\ref{subsec:fptas_proofs}. 
In this result, we use this simple observation that by moving a node $v$ from its initial partition to another partition, the cut value reduces by at most $d^{-}(v):=\sum_{u, \ell_u\neq \ell_v} c_{uv}$. Let $C_{0}$ denote the initial cut-set. Let $C_{\OPT}$ denote the optimal cut-set for the \maxmove~problem and  $S_{\OPT}$ be the set of nodes that have to be moved to partitions different from their initial ones in order to obtain $C_{\OPT}$. By the problem definition, $|S_{\OPT}|\le r$. Let the value of the optimal cut be $c(C_{\OPT})$ and the value of the initial cut be $c(C_0)$. If $c(C_0)$ is already a $(1+\epsilon)$ approximation of $c(C_{\OPT})$, then we can simply output the initial cut as our solution. Otherwise, we argue that there exists a node $v\in S_{\OPT}$ for which $d^-(v)\ge \frac{\epsilon}{r(1+\epsilon)}c(C_0)$. Let $A$ be the set of nodes $u$ with $d^{-}(u)\ge \frac{\epsilon}{r(1+\epsilon)}c(C_0)$. We show that $|A|\le \frac{2r(1+\epsilon)}{\epsilon}$. This allows us to achieve a recursive algorithm as follows: For each node $v\in A$ and each partition $i\neq \ell_v$, move $v$ to partition $i$, contract $v$ with the terminal node $s_i$, then reduce the move parameter by $1$ and recurse on this instance. Contraction of node $v$ with the terminal node $s_i$ is to forbid the algorithm from moving node $v$ again. Since we make $(k-1)|A|$ instances of the $(r-1)$-move $k$-partitioning problem, we can show that the running time of our algorithm is polynomial in $n$ and $1/\epsilon$. Note that this algorithm is polynomial only if $r$ is constant. Theorem \ref{thm:m2approx} provides an approximation algorithm that is polynomial in $n$ and is independent of $r$.

\subsection{Proof of Theorem \ref{thm:fptas}}\label{subsec:fptas_proofs}

Let $\alpha = 1+\epsilon$. We solve an instance of the \maxmove~problem with at most $(k-1)\frac{2r\alpha}{\alpha-1}$ instances of the $(r-1)$-move $k$-partitioning problem. 

Recall that the weight of any edge $(u,v)\in E(G)$ is denoted by $c_{uv}$. 
Let the initial cut-set in graph $G$ be $C_0$, i.e., $C_0=\{(u,v)\in E(G) | \ell_u\neq \ell_v\}$, and let $c(C_0)$ denote its cut value. For each node $v\in V(G)$, let $d^-(v)$ denote the sum of the weight of all edges incident to $v$ that lie in the initial cut-set $C_0$ i.e., $d^-(v)=\sum_{(u,v)\in C_0} c_{uv}$. Observe that $c(C_0)=\frac{1}{2}\sum_{v\in V(G)} d^-(v)$, since the weight of each edge $(u,v)\in C_0$ appears in both $d^-(u)$ and $d^-(v)$. 
 Let $A$ be the set of nodes $v$ with $d^-(v)\ge \frac{\alpha-1}{r\alpha}c(C_0)$. Our FPTAS makes $k-1$ instances of $(r-1)$-move $k$-partitioning for each $v\in A$ as follows: For any $j\neq \ell_v$, let $G_{v,j}$ be the graph obtained from $G$ by contracting node $v$ with terminal $s_j$. After this contraction, for any node $u$, the weight of edge $(s_j,u)$ is increased by $c_{vu}$. Let set $\mathbf{C}=\{C_0\}$. The algorithm then recurses on graph $G_{v,j}$ with move parameter $r-1$ for all $v'\in A$ and $j\neq \ell_{v'}$, 
 and adds the output of this recursion to the set of cut-sets $\mathbf{C}$. Then the algorithm outputs the cut-set with the least value among all those in $\mathbf{C}$. Note that a small post processing is needed for the output as we need to undo the contractions which are at most $r$ many, see Algorithm~\ref{alg:fptas}.


\begin{algorithm}[h]
\caption{FPTAS for constant move parameter}\label{alg:fptas}
\begin{algorithmic}[1]
\Require graph $G$ with terminals $s_1,\ldots, s_k$, move parameter $r$ and $\alpha=1+\epsilon$.
   \State $C_0=\{(u,v)\in E(G) | \ell_u\neq \ell_v\}\gets$ initial cut-set in graph $G$. 
    \State $c(C_0)\gets \sum_{u,v \text{ such that } (u,v)\in C_0} c_{uv}$. \Comment{Cut value of $C_0$.}
    \If{$r=0$}
        \State \Return $C_0$. \Comment{Return the initial cut, without moving anyone.}
    \Else

    \State $\mathbf{C}=\{C_0\}.$
    \For{$v\in V(G)$} 
        \State $d^-(v)=\sum_{u \text{ such that }(u,v)\in C_0} c_{uv}$. 
    \EndFor
      \State $A\gets\{v\in V(G)| d^-(v)\ge \frac{\alpha-1}{r\alpha}c(C_0)\}$.
    \For{$v\in A \cap V(G)$}
        \For{$j\in \{1,\ldots,k\}\setminus\{l_v\}$} \Comment{All partitions except $\ell_v$.}
            \State $G_{v,j} \gets$ Graph obtained from $G$ by contracting $v$ with $s_j$ and increasing the weight of each edge $(s_j,u)$ by $c_{vu}$.
            \State $C_{v,j}\gets$ The cut returned by running Algorithm \ref{alg:fptas} on $G_{v,j}$ with move parameter $r-1$.
            \State $\mathbf{C}\gets \mathbf{C}\cup C_{v,j}$.
        \EndFor
    \EndFor
    \State \Return $\arg\min_{C\in \mathbf{C}} c(C)$. \Comment{Return the cut with the smallest value in $\mathbf{C}$.}
    \EndIf
\end{algorithmic}
\end{algorithm}

We begin the analysis of this algorithm by showing that the running time of the algorithm is at most $(\frac{2(k-1)\alpha}{\alpha-1})^rr!O(n)$. Since $c(C_0)=\frac{1}{2}\sum_{v\in V(G)} d^-(v)$ we have  $c(C_0)\ge \frac{1}{2}\sum_{v\in A}d^-(v)\ge |A|\frac{\alpha-1}{2r\alpha}c(C_0)$; therefor, $|A|\le \frac{2r\alpha }{\alpha-1}$.
Let $F(r,n)$ be the running time of the algorithm on any $n$-node graph. We have $F(r,n)\le (k-1)\frac{2r\alpha }{\alpha-1}F(r-1,n-1) + O(n^2)$, where the second term is the time required to compute $d^-(v)$ for all $v\in V(G)$ . Note that the constant behind $n^2$ is independent of $r,k$ and $\alpha$. 
Since $F(0,n) = O(1)$ for any $n$, we have $F(r,n)\le (\frac{2(k-1)\alpha}{\alpha-1})^rr!O(n^2)$.

Now we prove that the algorithm outputs a cut-set $C'$ such that $c(C')\le \alpha c(C_{\OPT})$, where $C_{\OPT}$ is an optimal cut-set. Recall that $C'= \arg \max_{\bar{C}\in \mathbf{C}}c(\bar{C})$. Let $S_{\OPT}$ be the set of nodes that were moved from their initial partition to obtain $C_{\OPT}$. If $S_{\OPT}=\emptyset$, then $C_{\OPT}=C_0$. Since $C_0\in \mathbf{C}$, we have that $c(C_{\OPT})\le c(C')\le c(C_0)$, so $c(C')=c(C_{\OPT})$ and the problem is trivial. Therefore, suppose that $S_{\OPT}\neq\emptyset$. Here, two cases are possible: (i) $c(C_0)\le \alpha c(C_{\OPT})$, and (ii) $c(C_0)> \alpha c(C_{\OPT})$. In case (i), since $c(C')\le c(C_0)$, we have $c(C')\le \alpha c(C_{\OPT})$; thus, $C'$ is a valid $\alpha$ approximation of $C_{\OPT}$. For case (ii), we show that $S_{\OPT}\cap A\neq \emptyset$. This can be achieved by first showing that $c(C_{\OPT})\ge c(C_0)-\sum_{v\in S_{\OPT}}d^-(v)$. 

Let graph $G'$ be obtained by removing all nodes in $S_\OPT$ and edges incident to such nodes from graph $G$. Let $C'_0$ be the equivalent of initial cut $C_0$ on $G'$, that is, $C'_0 = C_0 \setminus \{(u,v)\in E(G) | v \in S_\OPT\}$. Similarly, we define $C'_\OPT$ as the portion of the optimal cut-set $C_\OPT$ that exists in $G'$, i.e., $C'_\OPT = C_\OPT \setminus \{(u,v)\in E(G) | v \in S_\OPT\}$. Since $C_0$ and $C_\OPT$ only differed in some edges incident to nodes in $S_{\OPT}$, we have $C'_0 = C'_\OPT$. As $c(C'_\OPT) \le c(C_\OPT)$, we have $c(C'_0) \le c(C_\OPT)$. On the other hand, since by removing each $v$ in $S_{\OPT}$ from graph $G$ the cut value drops by at most $d^-(v)$, we have that $c(C'_0)$ is at least $c(C_0)-\sum_{v\in S_{\OPT}}d^-(v)$. Putting these together provides $c(C_{\OPT})\ge c(C_0)-\sum_{v\in S_{\OPT}}d^-(v)$.  Recall that in case (ii), $c(C_0)> \alpha c(C_{\OPT})$; therefore, $\sum_{v\in S_{\OPT}}d^-(v)\ge c(C_0)-c(C_{\OPT})>\frac{\alpha-1}{\alpha}c(C_0)$. So there must be a node $v^*$ in $S_{\OPT}$ such that $d^-(v^*)>\frac{\alpha-1}{\alpha r}c(C_0)$, concluding that $S_{\OPT}\cap A\neq \emptyset$.

Suppose that in the optimal solution, node $v^*\in A$ is moved to partition $j$ for some $j\neq \ell_{v^*}$. We argue that the optimal cut value for $G_{v^*,j}$ when at most $r-1$ nodes can be moved to partitions different from their initial ones is at most $c(C_{\OPT})$. This is because contracting $v^*$ and $s_j$ in $G_{v^*,j}$ is equivalent to moving $v^*$ to partition $j$ and not moving it to another partition again. By moving the nodes in $S_{\OPT}\setminus \{v^*\}$ in order to obtain the cut $C_{OPT}$, the resulting cut value is $c(C_{\OPT})$. Let $C_{v^*,j}'$ be the output of the recursion on $G_{v^*,j}$. By induction, $c(C_{v^*,j}')\le \alpha c(C_{\OPT})$. Since $C_{v^*,j}'\in C$, the algorithm outputs an $\alpha$-approximation for $G$.

\section{$3(r+1)$-Approximation Algorithm for Parametric $r$}\label{sec:alg}
In this section, we show a $\frac{2k}{k-1}(r+1)$-approximation algorithm for the \maxmove~problem when $k>2$. For $k>2$, we have that $\frac{2k}{k-1}\le 3$; therefore, our algorithm provides a $3(r+1)$-approximation guarantee for such $k$. Note that, this algorithm works for $k=2$ as well; however, we already have an $(r+1)$-approximation algorithm for this case.

\begin{algorithm}[h]
\caption{$\frac{2k}{k-1}(r+1)$-approximation algorithm for the \maxmove~problem}\label{alg:r_approx}
\begin{algorithmic}[1]
   \State $\bar{X}\gets$ solution to LP~\ref{tab:maxmove-lp}. 
   \State $g = \frac{k-1}{k(r+1)}$.
    \State $\rho$ chosen uniformly at random from $(0,g)$.
    \State $A\gets\{\mathbf{z}=(\mathbf{z}^1,\ldots,\mathbf{z}^k)| \frac{1}{g}\mathbf{z}^i\in \mathbb{Z}_{\ge 0} \text{ for all } 1\le i\le k \text{ ~and~ } \sum_{i=1}^k \mathbf{z}^i < 1+kg \}$.
    \For{$i$ from $1$ to $k$}
        \For{$v\in V(G)$} 
            \State $\tilde{X}_v^i = g\floor{\frac{\bar{X}_v^i+\rho}{g}}.$
        \EndFor
    \EndFor
    \For{$\mathbf{z}\in A$}
        \State $H_\mathbf{z}\gets \{v\in V(G) | \tilde{X}_v = \mathbf{z}\}$.
        \If{$H_\mathbf{z}$ contains a terminal $s_i$} \Comment{$H_\mathbf{z}$ can contain at most one terminal}
            \State $i^*_\mathbf{z}\gets i$.
        \Else
            \State $i^*_\mathbf{z}\gets\arg \max_{j=1}^k |\{v\in H_\mathbf{z}|\ell_v=j\}|$.
        \EndIf
        \State Assign all the nodes in $H_\mathbf{z}$ to $i^*_\mathbf{z}$.
    \EndFor
\end{algorithmic}
\end{algorithm}

The algorithm operates as follows: The optimal solution to  LP~\ref{tab:maxmove-lp} gets rounded in two phases. Let $g = \frac{k-1}{k(r+1)}$ and $\rho$ be a real number chosen uniformly at random from the interval $(0,g)$. In the first phase, for each node $v$, we round the entries of $\bar{X}_v$ to obtain
$\tilde{X}_v$ such that for each $1\le i\le k$, $\tilde{X}_v^i$ is an integer multiple of $g$. More precisely, let $\tilde{X}_v^i = \floor{\frac{\bar{X}_v^i+\rho}{g}}g$, i.e., $\tilde{X}_v^i$ is the largest multiple of $g$ no larger than 
 $\bar{X}_v^i+\rho$. 

In the second phase of rounding, the algorithm puts all nodes that are rounded to the same vector in the same partition, see Algorithm~\ref{alg:r_approx} for a description of how this partition is chosen. 
To clarify the role of $\rho$, for two nodes $u,v\in V(G)$ if $d_{\bar{X}}(u,v)$ is small, then we want $\bar{X}_u$ and $\bar{X}_v$ be rounded to the same vector with high probability, i.e., $\tilde{X}_u=\tilde{X}_v$ with high probability. If we let $\rho=0$ at all times, then $u$ and $v$ never get rounded to the same vector in the following case: For some small $\epsilon>0$, let $\bar{X}_v^1=g+\epsilon=\bar{X}_u^2$, $\bar{X}_v^2=g-\epsilon=\bar{X}_u^1$ and $\bar{X}_v^i=\bar{X}_u^i$ for $i>2$. Thus, $d_{\bar{X}}(u,v)= 2\epsilon$ is small. For this example and $\rho=0$ we have $\tilde{X}_v^1= \floor{\frac{\bar{X}_v^1}{g}}g = \floor{\frac{g+\epsilon}{g}}g = g$ and $\tilde{X}_u^1 = \floor{\frac{\bar{X}_u^1}{g}}g = \floor{\frac{g-\epsilon}{g}}g = 0$. In other words, even though $d_{\bar{X}}(u,v)$ is small $\bar{X}_u$ and $\bar{X}_v$ are rounded to different vectors with certainty. This is in fact true for any constant value of $\rho$ and a random $\rho$ is key to be able to round two nodes $u$ and $v$ that are very close to each other to the same vector, with high probability.

\begin{lemma}\label{lem:rounding1}
    For each $v\in V(G)$ and each $i\in \{1,\ldots,k\}$ we have  $|\tilde{X}_v^i-\bar{X}_v^i|<g$. Moreover, if for two nodes $u,v\in V(G)$ and an $i\in \{1,\ldots,k\}$ we have $\tilde{X}_v^i = \tilde{X}_u^i$, then $|\bar{X}_v^i -\bar{X}_u^i|<g$.
\end{lemma}
\begin{proof}
    If $\tilde{X}_v^i\ge\bar{X}_v^i$, then since $\tilde{X}_v^i\le \bar{X}_v^i+\rho$ 
    and $\rho<g$, we have that $|\tilde{X}_v^i-\bar{X}_v^i|<g$. If $\tilde{X}_v^i<\bar{X}_v^i$, then since $\tilde{X}_v^i = \floor{\frac{\bar{X}_v^i+\rho}{g}}g$, we have $(\frac{\bar{X}_v^i+\rho}{g}-1)g<\tilde{X}_v^i$, thus $\bar{X}_v^i+\rho<\tilde{X}_v^i + g$. Therefore, $|\tilde{X}_v^i-\bar{X}_v^i|<g-\rho<g$.

    For the second part of the lemma, since both $\bar{X}_v^i+\rho$ and $\bar{X}_u^i+\rho$ are rounded down to the same value $\alpha = \tilde{X}_v^i$, they are both in the $[\alpha,\alpha+g)$ interval. Therefore, $\bar{X}_v^i$ and $\bar{X}_u^i$ are both in the interval $[\alpha-\rho,\alpha+g-\rho)$, concluding that $|\bar{X}_v^i -\bar{X}_u^i|<g$.
\end{proof}

By Lemma \ref{lem:rounding1}, we have $|\tilde{X}_v^i-\bar{X}_v^i|<g$ and by constraint (C4) in LP~\ref{tab:maxmove-lp} we have $\sum_{i=1}^k \bar{X}_v^i=1$, thus, $\sum_{i=1}^k \tilde{X}_v^i < 1+kg$. Let $A$ be the set of all $k$-sized vectors with entries being non-negative integer multiples of $g$ such that the sum of the entries of each vector is at most $1+kg$. By definition, each entry of $\tilde{X}_v$ is non-negative integer multiple of $g$ and, as shown, the sum of entries of $\tilde{X}_v$ is at most $1+kg$. Therefore, Algorithm~\ref{alg:r_approx} rounds any $\bar{X}_v$ to a vector in $A$. 


Now we show that if $d_{\bar{X}}(u,v)$ is very small, then with high probability $\tilde{X}_u = \tilde{X}_v$. Note that for two nodes $u$ and $v$ with $\tilde{X}_u = \tilde{X}_v$, Algorithm~\ref{alg:r_approx} puts $u$ and $v$ in the same partition.
\begin{lemma}\label{lem:cut_probs}
    For any two distinct nodes $u,v\in V(G)$ such that $(u,v)\in E(G)$, the probability that Algorithm~\ref{alg:r_approx} puts $(u,v)$ in the cut-set is at most $\frac{2}{g}d_{\bar{X}}(u,v)$.
\end{lemma}
\begin{proof}
If $d_{\bar{X}}(u,v)>\frac{g}{2}$, then $1<\frac{2}{g}d_{\bar{X}}(u,v)$ and the lemma holds trivially. Assume that $d_{\bar{X}}(u,v)\le \frac{g}{2}$. 
Observe that $u$ and $v$ are assigned to different partitions only if they are rounded to different vectors, that is, $\tilde{X}_u \neq \tilde{X}_v$. This happens if there exists an $i\in \{1,\ldots,k\}$ and an integer $1\le t\le \frac{1}{g}$ such that $tg$ is between $\bar{X}_v^i+\rho$ and $\bar{X}_u^i+\rho$. To measure this probability, let $i\in \{1,\ldots,k\}$ be fixed. Without loss of generality, suppose that  $\bar{X}_v^i< \bar{X}_u^i$. Then, we need to compute the probability that
\begin{equation}\label{eq:in-between}
\bar{X}_v^i+\rho<tg\le \bar{X}_u^i+\rho. 
\end{equation}
%
To see feasible values of $\rho$ and $t$, we consider the following two cases: (i) $\floor{\frac{\bar{X}_v^i}{g}} = \floor{\frac{\bar{X}_u^i}{g}}$, and (ii) $\floor{\frac{\bar{X}_v^i}{g}} <\floor{\frac{\bar{X}_u^i}{g}}$. For the first case, suppose that $\floor{\frac{\bar{X}_v^i}{g}} = \floor{\frac{\bar{X}_u^i}{g}} = \hat{t}$. This means that $\hat{t}g\le \bar{X}_v^i< \bar{X}_u^i <(\hat{t}+1)g$. Since $\rho < g$, the only feasible value for $t$ is $\hat{t}+1$ and inequality \eqref{eq:in-between} is equivalent to \begin{equation}\label{eq:const1}
0<(\hat{t}+1)g-\bar{X}_u^i\le \rho< (\hat{t}+1)g-\bar{X}_v^i<g.
\end{equation} 
The length of the $[(\hat{t}+1)g-\bar{X}_u^i, (\hat{t}+1)g-\bar{X}_v^i)]$ interval 
is $|\bar{X}_u^i-\bar{X}_v^i|$. Therefore, in case (i), with probability at most $\frac{1}{g}|\bar{X}_u^i-\bar{X}_v^i|$ there exists a $t$ satisfying inequality \eqref{eq:in-between}. 

For the second case, suppose that $\hat{t} = \floor{\frac{\bar{X}_v^i}{g}} <\floor{\frac{\bar{X}_u^i}{g}}$. Since $g>\frac{g}{2}>d_{\bar{X}}(u,v)\ge \bar{X}_u^i-\bar{X}_v^i$, we must have $\floor{\frac{\bar{X}_u^i}{g}} = \hat{t}+1$. The only accepted values for $t$ is $\hat{t}+1$ and $\hat{t}+2$, and inequality \eqref{eq:in-between} is satisfied when $\rho$ is in one of the following two intervals:
$$
0\le\rho<(\hat{t}+1)g - \bar{X}_v^i  \text{~~ or ~~} (\hat{t}+2)g-\bar{X}_u^i\le\rho \le g. 
$$
The sum of the lengths of these two intervals is $|\bar{X}_u^i-\bar{X}_v^i|$, so in case (ii) edge $(u,v)$ is in the cut-set with probability at most $\frac{1}{g}|\bar{X}_u^i-\bar{X}_v^i|$. Putting cases (i) and (ii) together, any edge $(u,v)\in E(G)$ is in the cut-set with probability at most $\frac{1}{g}|\bar{X}_u^i-\bar{X}_v^i|$. Finally, since $d(u,v)=\frac{1}{2}\sum_{i=1}^k |\bar{X}_u^i-\bar{X}_v^i|$, the probability that edge $(u,v)$ is in the cut-set is at most $\frac{2}{g}d_{\bar{X}}(u,v)$.
\end{proof}

Let $\mathbf{z}$ be a vector in $A$ and $\mathbf{z}^i$ denote its $i-$th entry. Let $H_\mathbf{z}$ be the set of all nodes $v\in V(G)$ for which $\tilde{X}_v=\mathbf{z}$. The following lemma discusses the number of terminal nodes that $H_\mathbf{z}$ may contain.

\begin{lemma}
\label{lem:max1terminal}
 For any vector $\mathbf{z}\in A$, there can be at most one terminal node in $H_\mathbf{z}$. 
\end{lemma}
\begin{proof}
Suppose there exists two distinct terminal nodes $s_i$ and $s_j$ in $H_\mathbf{z}$. By constraint (C5) of LP~\ref{tab:maxmove-lp}, we have $\bar{X}_{s_i}^i=1$ and $\bar{X}^i_{s_j}=0$, providing $|\bar{X}_{s_i}^i-\bar{X}^i_{s_j}|=1$. On the other hand, by Lemma \ref{lem:rounding1} we have $|\bar{X}_{s_i}^i-\bar{X}^i_{s_j}|<g<1$. This is a contradiction and there can be at most one terminal in $H_\mathbf{z}$.
\end{proof}

If there exists a terminal node $s_i\in H_\mathbf{z}$, then Algorithm~\ref{alg:r_approx} puts all the nodes in $H_\mathbf{z}$ in partition $i$. Otherwise, the algorithm puts all the nodes in $H_\mathbf{z}$ in the partition where most nodes of $H_\mathbf{z}$ come from. Let $r_\mathbf{z} = \sum_{v\in H_\mathbf{z}}(1-\bar{X}_v^{\ell_v})$.

\begin{lemma}\label{lem:rmoves}
If $r\ge 1$, then for each vector $\mathbf{z}\in A$, there exists a partition $i_\mathbf{z}$ for which the number of nodes $v\in H_\mathbf{z}$ with $\ell_v\neq i_\mathbf{z}$ is less than $r_\mathbf{z}+\frac{r_\mathbf{z}}{r}$. If there is a terminal node $s_i$ in $H_\mathbf{z}$, then $i_\mathbf{z}=i$.
\end{lemma}
\begin{proof} 
Consider a vector $\mathbf{z}\in A$ and its associated set of nodes $H_\mathbf{z}$. Let $B_i$ be the set of nodes $v\in H_\mathbf{z}$ with $\ell_v= i$, i.e., the set of nodes in $H_\mathbf{z}$ whose initial partition is $i$. Similarly, we define $B_{-i}$ to be the set of nodes $v\in H_\mathbf{z}$ with $\ell_v\neq i$, that is, the set of nodes in $H_\mathbf{z}$ whose initial partition is not $i$. To prove this lemma, we consider the following two cases: (i) there exists a terminal node $s_i\in H_\mathbf{z}$, and (ii) $H_\mathbf{z}$ contains no terminal node. 

In case (i), 
for a node $v\in H_\mathbf{z}$, it holds that $\tilde{X}_v^i = \tilde{X}_{s_i}^i$.
So, by Lemma \ref{lem:rounding1}, we have $\bar{X}_v^i>\bar{X}_{s_i}^i-g=1-g$. Using the fact that for any $v\in V(G)$ it holds that $\sum_{i\in\{1,\ldots,k\}} \bar{X}_v^i = 1$, we have 
$r_\mathbf{z} =  \sum_{v\in H_\mathbf{z}}(1-\bar{X}_v^{\ell_v}) \ge \sum_{v\in B_{-i}}(1-\bar{X}_v^{\ell_v}) \ge \sum_{v\in B_{-i}}\bar{X}_v^i> |B_{-i}|(1-g),$
providing that $|B_{-i}|< \frac{1}{1-g}r_\mathbf{z}$. Combining this with $g = \frac{k-1}{k(r+1)}<\frac{1}{1+r}$, we have $|B_{-i}|<r_\mathbf{z}+\frac{r_\mathbf{z}}{r}$.

In case (ii), $H_\mathbf{z}$ contains no terminal node. 
Without loss of generality, suppose that for all $2< i \le k$ we have $|B_1|\le |B_2|\le |B_i|$. Defining $B' = H_\mathbf{z}\setminus (B_1\cup B_2)$, we have 
    $|B_1|\le \frac{|B'|}{k-2}.$
For the sake of contradiction, suppose for all $i\in \{1,\ldots,k\}$ we have $|B_{-i}|=|H_\mathbf{z}\setminus B_i|\ge r_\mathbf{z}+\frac{r_\mathbf{z}}{r}$. In particular, $r_\mathbf{z}+\frac{r_\mathbf{z}}{r}\le |H_\mathbf{z}\setminus B_2| = |B'|+|B_1|\le |B'|+\frac{|B'|}{k-2}$. 
So $|B'|\ge \frac{k-2}{k-1}(r_\mathbf{z}+\frac{r_\mathbf{z}}{r})$. By Lemma~\ref{lem:rounding1}, for any two nodes $v,u\in H_\mathbf{z}$, we have $
|\bar{X}_v^{\ell_v}-\bar{X}_{u}^{\ell_v}|<g$; thus, $1-\bar{X}_v^{\ell_v}>1-\bar{X}_{u}^{\ell_v}-g$. 
If $\ell_v\notin \{1,2\}$, that is, $v\in B'$, then  $\sum_{i\in\{1,\ldots,k\},i\neq \ell_v}\bar{X}_v^i > g+\sum_{i\in \{1,2\}, i\neq \ell_v}(\bar{X}_{u}^i-g) + \sum_{i\in\{3,\ldots,k\},i\neq \ell_v}\bar{X}_{u}^i$. If $\ell_v \in \{1,2\}$, then we have $\sum_{i\in\{1,\ldots,k\},i\neq \ell_v}\bar{X}_v^i >\sum_{i\in \{1,2\}, i\neq \ell_v}(\bar{X}_{u}^i-g) + \sum_{i\in\{3,\ldots,k\},i\neq \ell_v}\bar{X}_{u}^i$. Therefore,
%
%
%
    $r_\mathbf{z} = \sum_{v\in H_\mathbf{z}} (1-{\bar{X}}_{v}^{\ell_v}) = \sum_{v\in H_\mathbf{z}} \sum_{i\in\{1,\ldots,k\},i\neq \ell_v} {\bar{X}}_v^i$ and so   $r_z>|B'|g+\sum_{v\in H_\mathbf{z}}\bigg[ \sum_{i\in \{1,2\}, i\neq \ell_v}(\bar{X}_{u}^i-g) \notag + \sum_{i\in\{3,\ldots,k\},i\neq \ell_v}\bar{X}_{u}^i\bigg]$
\begin{align}
    \notag r_z&> |B'|g+\sum_{i\in \{1,2\}} \sum_{v\in H_\mathbf{z}, \ell_v\neq i} (\bar{X}_{u}^i-g) \notag \\ &+ \sum_{i\in\{3,\ldots,k\}} \sum_{v\in H_\mathbf{z}, \ell_v\neq i} \bar{X}_{u}^i\notag \\
    &\ge \frac{k-2}{k-1}(r_\mathbf{z}+\frac{r_\mathbf{z}}{r})g+ \sum_{i\in \{1,2\}}(r_\mathbf{z}+\frac{r_\mathbf{z}}{r})(\bar{X}_{u}^i-g) \notag \\ &+ \sum_{i\in\{3,\ldots,k\}} (r_\mathbf{z}+\frac{r_\mathbf{z}}{r}) \bar{X}_{u}^i\label{eq:contradiction}\\
    & = (r_\mathbf{z}+\frac{r_\mathbf{z}}{r})(1-\frac{k}{k-1}g)= r_\mathbf{z}\notag,
\end{align}
where inequality~\eqref{eq:contradiction} uses the assumption that $|B_{-i}|\ge r_\mathbf{z}+\frac{r_\mathbf{z}}{r}$ for all $i\in \{1,\ldots,k\}$. The above calculations reaches a contradiction, $r_\mathbf{z} > r_\mathbf{z}$. Therefore, there exists an $i_\mathbf{z}\in \{1,\ldots, k\}$ such that the number of nodes $v\in H_\mathbf{z}$ with $\ell_v\neq i_\mathbf{z}$ is at most $r_\mathbf{z}+\frac{r_\mathbf{z}}{r}$, concluding the proof of the lemma. 
\end{proof}


\begin{proofof}{Theorem \ref{thm:m2approx}}
    Algorithm \ref{alg:r_approx} sets $g = \frac{k-1}{k(r+1)}$. Firstly, note that by Lemma \ref{lem:cut_probs} the probability that an edge $(u,v)$ is in the cut-set is at most $\frac{2k}{k-1}(r+1)d_{\bar{X}}(u,v)$, where $\bar{X}$ is an optimal solution to LP~\ref{tab:maxmove-lp}. Therefore, Algorithm \ref{alg:r_approx} returns a cut-set with cut value at most $\frac{2k}{k-1}(r+1)\OPT_r$, where $\OPT_r$ denotes the optimal value of LP~\ref{tab:maxmove-lp} with a move parameter of $r$. 
    
    Secondly, we prove that Algorithm \ref{alg:r_approx} moves at most $r$ nodes. To do this, we show that for each vector $\mathbf{z}\in A$, the number of nodes in $H_\mathbf{z}$ that are moved from their initial partitions is less than $r_\mathbf{z}+\frac{r_\mathbf{z}}{r}$. There are two possibilities to consider here. Case (i): There exists a terminal in $s_i \in H_\mathbf{z}$ (recall that by Lemma \ref{lem:max1terminal}, there can be at most one terminal in $H_\mathbf{z}$). In this case, Algorithm \ref{alg:r_approx} assigns all the nodes in $H_\mathbf{z}$ to partition $i$. Therefore, the number of nodes that are moved to partitions different from their original partitions is equal to the number of nodes $v\in H_\mathbf{z}$ with $\ell_v\neq i$. By Lemma \ref{lem:rmoves}, the number of such nodes is less than $r_\mathbf{z}+\frac{r_\mathbf{z}}{r}$. Case (ii): $H_\mathbf{z}$ has no terminal nodes. In this case, Algorithm \ref{alg:r_approx} assigns all the nodes to partition $i^*_\mathbf{z}\gets\arg \max_{j=1}^k |\{v\in H_\mathbf{z}|\ell_v=j\}|$. Defining $B_i$ as the set of nodes $v\in H_\mathbf{z}$ with $\ell_v = i$, the number of nodes that are moved to  partitions different from their original ones is equal to $|H_\mathbf{z}|-|B_{i^*_\mathbf{z}}|$. By Lemma \ref{lem:rmoves}, there exists a partition $i$ for which $|H_\mathbf{z}|-|B_i|<r_\mathbf{z}+\frac{r_\mathbf{z}}{r}$. By definition of $i^*_\mathbf{z}$, we have $|B_i|\le |B_{i^*_\mathbf{z}}|$; therefore, $|H_\mathbf{z}|-|B_{i^*_\mathbf{z}}|<r_\mathbf{z}+\frac{r_\mathbf{z}}{r}$. 

    The total number of nodes that are moved by Algorithm~\ref{alg:r_approx} in cases (i) and (ii) is strictly less than $\sum_{\mathbf{z}\in A} r_\mathbf{z}+\frac{r_\mathbf{z}}{r}$. Using the definition of $r_\mathbf{z}$ and constraint (C7) of LP~\ref{tab:maxmove-lp}, we have $\sum_{\mathbf{z}\in A} r_\mathbf{z}+\frac{r_\mathbf{z}}{r} = (\frac{r+1}{r})\sum_{\mathbf{z}\in A} r_\mathbf{z} = (\frac{r+1}{r})\sum_{\mathbf{z}\in A} \sum_{v\in H_\mathbf{z}} (1-\bar{X}^{\ell_v}_v) = (\frac{r+1}{r}) \sum_{v\in V(G)} (1-\bar{X}^{\ell_v}_v) \le r+1$. This shows that the total number of nodes that Algorithm~\ref{alg:r_approx} moves is strictly less than $r+1$, in other words, at most $r$ nodes are moved to partitions different from their original ones by the algorithm. 
\end{proofof}

\subsection{Run Time and De-randomization}\label{sec:derandom}
For computing the run time, note that finding the solution to LP\ref{tab:maxmove-lp} takes $T_{LP}(n,m)$ time. The rest of the algorithm consists of simple loops and takes $O(mk)$ time.

One can de-randomize the algorithm using the technique of Calinescu, Karloff and Rabani \cite{LP_paper}. They de-randomize their rounding for their $1.5-1/k$ approximation algorithm for the multiway-cut problem (at the expense of polynomial increase in running time), and they argue that there are $O(n)$ “interesting” values of $\rho$ and hence one only needs to run the rounding algorithm for those values. To see this in our algorithm at a high level, note that for each $i$, there is a value $\rho_i$ where for all $\rho<\rho_i$, $\tilde{X}_v^i$ is $g\floor{\frac{\bar{X}_v^i}{g}}$ and for $\rho\ge \rho_i$, we have $\tilde{X}_v^i = g\floor{\frac{\bar{X}_v^i}{g}+1}$. The value of $\rho_i$ can be computed from $\bar{X}_v^i$. So one can only run the algorithm for these values of $\rho_i$. For the running time, it increases by a factor of at most $n$, however one could decrease this factor by sorting the values of $\rho$. 





\section{Constant-Factor Bicriteria Approximation Algorithm}\label{sec:bich}

To develop our constant-factor bicriteria approximation algorithm, we use some of the techniques used in~\cite{LP_paper}. We first provide a detailed review on the rounding LP~\ref{tab:paper-lp} technique of \cite{LP_paper}.
\subsection{Review of the Rounding Technique of LP \ref{tab:paper-lp}}\label{subsucsec:rounding} 
In this section, we explain the $1.5$-approximation rounding of \cite{LP_paper}. 
Let an $n$-node graph $G$ be the instance of the Multiway cut problem and $X$ be a solution to the Multiway cut CKR LP (LP \ref{tab:paper-lp}). 
We first add some nodes to graph $G$ via edge subdivision and then extend solution $X$ to the added nodes such that for each edge $(u,v)$ in the new graph, $X_u$ and $X_v$ differ in at most two entries. This edge subdivision mechanism is discussed below and its steps are shown in Algorithm~\ref{alg:edgesubdiv}.

\begin{table*}[hht]
    \centering
    \begin{tabular}{lllr}
    \toprule
        \text{Minimize:} & $\sum_{(u,v)\in E}c_{uv}d(u,v)$ & &\\
        \text{Subject to:} & & &\\
        & $d(u,v)= \frac{1}{2}\sum_{i=1}^k Y^i_{uv}$  & $\forall  (u,v)\in E$, $\forall i\in\{1\cdots k\}$  & (C1)\\
        & $Y^i_{uv}\ge X^i_u-X^i_v$ & $\forall (u,v)\in E$, $\forall i\in\{1,\ldots, k\}$& (C2)\\
        & $Y^i_{uv}\ge X^i_v-X^i_u$ & $\forall (u,v)\in E$, $\forall i\in\{1,\ldots, k\}$& (C3)\\
        & $\sum_{i=1}^{k}X^i_u=1$ & $\forall u\in V$, $\forall i\in\{1,\ldots, k\}$ & (C4)\\
        & $X_{s_i}=\mathbf{e}_i$ & $\forall s_i\in S$  & (C5)\\
        & $X_v^i \ge 0$ & $\forall v\in V$, $\forall i\in \{1,\ldots,k\}$ & (C6)\\
        \bottomrule
    \end{tabular}
    \caption{Multiway cut CKR LP (LP~\ref{tab:paper-lp})}
    \vspace{-0.3cm}
    \label{tab:paper-lp}
\end{table*}

Consider two nodes $u$ and $v$, where $X_u$ and $X_v$ differ in more than two entries. Subdivide the edge $(u,v)$ using a new node $w$ and let the weight of edges $(u,w)$ and $(w,v)$ be the same as the weight of the deleted edge $(u,v)$. Define $X_w$ as follows: Let $i$ be the entry that holds the minimum non-zero value among $|X_u^t-X_v^t|$ for all entries $t\in\{1,\ldots,k\}$. Without loss of generality, assume that $X_u^i<X_v^i$ and let $\alpha = X_v^i-X_u^i$. Since $\sum_{t=1}^k X_u^t= \sum_{t=1}^k X_v^t =1$, there exists an entry $j\neq i$ for which $X_u^j-X_v^j \ge \alpha$. We define $X_w^i = X_u^i+\alpha=X_v^i$, $X_w^j = X_u^j-\alpha$. Moreover, for any entry $t\neq i,j$, $X_w^t$ is defined as $X_w^t = X_u^t$; therefore, nodes $u$ and $w$ differ in only two entries. Observe that $d(u,v)=d(u,w)+d(w,v)$. Suppose $u$ and $v$ are different in $a$ number of entries, then $v$ and $w$ are different in $a-1$ number of entries. Graph $G$ has $n$ nodes and by going through the explained procedure the number of different entries between any two nodes decreases by one. Each two nodes in graph $G$ can have at most $k$ different entries. Therefore, after $O(nk)$ steps this procedure ends. The resulting graph by the end of this procedure is denoted by $G^*=(V^*,E^*)$ and the vectors of endpoints of each edge of it are different in at most two entries, that is, for any edge $(u,v)\in E^*$, $X_u$ and $X_v$ differ in at most two entries. Note that since $d(u,v)=d(u,w)+d(w,v)$ and $c_{uv}=c_{wu}=c_{wv}$, the values of optimal solutions of LP~\ref{tab:paper-lp} on graphs $G$ and $G^*$ are the the same.

Using Algorithm~\ref{alg:edgesubdiv}, without loss of generality, it can be assumed that for any edge $(u,v)$ of an instance of the Multiway cut problem $G$, $\bar{X}_u$ and $\bar{X}_v$ differ in at most two entries, where $\bar{X}$ is an optimal solution to LP \ref{tab:paper-lp} on graph $G$. Let $\rho$ be a parameter chosen uniformly at random from [0,1]. Let $B(\rho,i)$ be the set of nodes $v\in V(G)$ for which $\bar{X}_v^{i}>\rho$. 
Let $\sigma=(\sigma(1),\ldots,\sigma(k))$ be either $(1,2,\ldots,k-1,k)$ or $(k-1,k-2,\ldots,1,k)$ with equal probabilities $1/2$. 
The partitions are processed in the order $\sigma(1),\sigma(2),\ldots, \sigma(k)$. In the first step, all nodes in $B(\rho,\sigma(1))$ are placed in partition $\sigma(1)$. Note that for any $i\neq j$, there might be some nodes that are in both $B(\rho,i)$ and $B(\rho,j)$. For each $j$ from $2$ to $k-1$, the remaining nodes in $B(\rho, \sigma(j))$ are placed in partition $\sigma(j)$, that is, all the nodes in $B(\rho, \sigma(j))\setminus \cup_{i=1}^{j-1} B(\rho, \sigma(i))$. All the remaining nodes  are then placed in partition $k$. Let $C$ be the cut-set made by the partitions produced by this rounding process. For an edge $(u,v)\in E(G)$, we have $\Pr((u,v) \in C)$ is the probability that edge $(u,v)$ is in cut-set $C$.

\begin{lemma}
[\cite{LP_paper}] For any edge $(u,v)\in E(G)$, we have $\Pr((u,v)\in C)\le 1.5\times d(u,v)$. 
\end{lemma}
\begin{proof}
Suppose that $\bar{X}_u$ and $\bar{X}_v$ differ in entries $i$ and $j$. Without loss of generality, assume that $\bar{X}_u^i\ge \bar{X}_v^i,\bar{X}_u^j,\bar{X}_v^j$. Since $\sum_{t=1}^k \bar{X}_u^t=\sum_{t=1}^k \bar{X}_v^t=1$, we must have that either $\bar{X}_u^i\ge \bar{X}_v^j\ge \bar{X}_v^i\ge \bar{X}_u^j$ or $\bar{X}_u^i\ge \bar{X}_v^i\ge \bar{X}_v^j\ge \bar{X}_u^j$, thus, $\bar{X}_u^j\le \bar{X}_u^i,\bar{X}_v^i,\bar{X}_v^j$.

In order for edge $(u,v)$ to be in cut-set $C$, one of its endpoints must belong in partition $i$ and the other in partition $j$, not any other partitions. This is because for any $\ell\neq i,j$, $\bar{X}_v^\ell=\bar{X}_u^\ell$ and for the chosen $\rho$, either both $u$ and $v$ are in $B(\rho,\ell)$ or neither are in it. Let $L_i$ (and similarly $L_j$) be defined as the interval between values $\bar{X}_v^i$ and $\bar{X}_u^i$, that is,  $L_i=(\bar{X}_v^i,\bar{X}_u^i)$ and $L_j=(\bar{X}_u^j,\bar{X}_v^j)$. Edge $(u,v)$ is in cut-set $C$ only if either $\rho\in L_i$ or $\rho \in L_j$. By constraint (C4) of LP~\ref{tab:paper-lp}, both intervals $L_i$ and $L_j $ are of the same length. However, $L_i$ and $L_j $ might be overlapping; thus, $|L_i\cup L_j|\le 2|L_i| = 2d_{\bar{X}}(u,v)$, where $d_{\bar{X}}(u,v)=\frac{1}{2}\big(|\bar{X}_u^i-\bar{X}_v^i|+|\bar{X}_u^j-\bar{X}_v^j|\big)$.

With equal probabilities, the aforementioned partitioning method processes partition $j$ before $i$. We claim that $(u,v)$ is going to be in $C$ if $\rho\in L_i\cup L_j$. To see that, consider the following three scenarios: (i) If partition $j$ is processed before partition $i$ and $\rho\in L_j$, then $v$ goes to partition $j$ but $u$ does not go to partition $j$. (ii) If partition $j$ is processed before partition $i$ and $\rho\in L_i\setminus L_j$, then we have that $\bar{X}_v^j,\bar{X}_u^j<\rho$, so $u$ and $v$ are not assigned to partition $j$. In this case, when partition $i$ is processed $u$ is assigned to $i$ and $v$ is not. (iii) If $i$ is processed before $j$ and $\rho\in L_i$, then $u$ is assigned to $i$ and $v$ is not. Note that if $\rho\in L_j\setminus L_i$, both $u$ and $v$ are assigned to $i$. This is because $\rho<\bar{X}_u^i,\bar{X}_v^i$ since $\rho\notin L_i$ and all values in $L_j\setminus L_i$ are less than the values in $L_i$. Using the union bound, the probability that edge $(u,v)$ is in cut-set $C$ is as follows:
\begin{align*}
    \Pr\big((u,v)\in C\big) &\le \frac{1}{2}\Pr(\rho\in L_i\cup L_j) + \frac{1}{2}\Pr(\rho\in L_i) \\&\le \frac{1}{2}\times 2d_{\bar{X}}(u,v)+\frac{1}{2}d_{\bar{X}}(u,v)\\&\le \frac{3}{2}d_{\bar{X}}(u,v);
\end{align*}
in other words, the cut value is $1.5$ times the objective value of the optimal solution to LP \ref{tab:paper-lp}, in expectation.   
\end{proof}


\subsection{Bicriteria Algorithm}
We describe our bicriteria algorithm in the next theorem.
%

\fourmove*

\begin{algorithm}[h]
\caption{Edge subdivision algorithm \cite{LP_paper}}\label{alg:edgesubdiv}
\begin{algorithmic}[1]
    \Require Input $X$ which is a solution to the  LP~\ref{tab:paper-lp}.
    \State $G^*\gets G$.
    \While{there is an edge $(u,v)\in E(G^*)$ such that $X_u$ and $X_v$ differ in more than two entries}
        \State Subdivide edge $(u,v)$ with new a node $w\in G^*$.
        \State $c_{uw}\gets c_{uv}$, $c_{vw}\gets c_{uv}$.
        \State $i\gets \arg \min_{t\in\{1,\ldots,k\}, |X_u^t-X_v^t| \neq 0} |X_u^t-X_v^t|$. 
        \If{$X_u^i>X_v^i$}
            \State Swap $u$ and $v$.
        \EndIf
        \State $\alpha \gets X_v^i-X_u^i$. Let $j$ be an entry where $X_u^j-X_v^j\ge \alpha$.
        \State $X_w^i\gets X_u^i+\alpha$, $X_w^j\gets X_u^j-\alpha$, $X_w^t=X_u^t$ for all $t\neq i,j$.
    \EndWhile
    \State \Return $G^*$ and the extension of $X$ on it.
\end{algorithmic}
\end{algorithm}

\begin{algorithm}[h]
\caption{Bicriteria approximation algorithm
}\label{alg:4Mmove}
\begin{algorithmic}[1]
    \Require Input graph $G$ with an initial partitioning, parameter $\gamma>\frac{1}{2}$ and move parameter $r$.  
    \State $\bar{X}\gets$ an optimal solution to the LP~\ref{tab:maxmove-lp}.
    
    \State Form graph $G^*$ as a super graph of $G$ and extend $\bar{X}$ on $G^*$ by Algorithm \ref{alg:edgesubdiv}. 
    \State Choose $\lambda$ uniformly at random from $[\frac{\gamma+1}{3}, \gamma]$. 
    \State Choose $\rho$ uniformly at random from $[0, \lambda]$. 
    \State $S= \emptyset$
    \For{all nodes $v\in G^*$}
        \If{there exists $i\in \{1,\ldots,k\}$ such that $\bar{X}_v^i\ge \lambda$ }
            \State $Z_v\gets \mathbf{e}_i$.
            \State $S= S\cup \{v\}$.
        \EndIf
    \EndFor
    \State $B(\rho,i)\gets$ the set of nodes $v$ not in $S$ such that $\bar{X}_v^i>\rho$, for each $i\in \{1,\ldots,k\}$. 
    \State Choose $\sigma$ to be one of the permutations $(1,2,\ldots, k-1,k)$ and $(k-1,\ldots,1,k)$ with equal probabilities.
    \For{each $j< k$}
        \For{all nodes $v\notin S$ that are in $B(\rho, \sigma(j))\setminus \cup_{i=1}^{j-1} B(\rho,\sigma(i))$}
            \State $Z_v \gets \mathbf{e}_{\sigma(j)}$.
        \EndFor
    \EndFor
    \For{any node $v$ not assigned to a partition}
        \State $Z_v \gets \mathbf{e}_k$.
    \EndFor
    \State \Return $Z$.
\end{algorithmic}
\end{algorithm}
\begin{proof}
The first step of the bicriteria algorithm is to solve LP \ref{tab:maxmove-lp}. Let $\bar{X}$ be an optimal solution to this LP. 
Similar to the rounding of LP \ref{tab:paper-lp} by \cite{LP_paper}, we can form a graph $G^*$ and extend the solution $\bar{X}$ on it, where $G^*$ is formed by subdividing the edges of $G$ such that for every two nodes $u,v$ that are adjacent, $\bar{X}_u$ and $\bar{X}_v$ differ in at most two entries. Note that the values of the optimal cuts in graphs $G$ and $G^*$ are the same. The subdivision process is depicted in Algorithm \ref{alg:edgesubdiv}. Remark that the move constraint of at most $r$ nodes is only on the original nodes, i.e., nodes of graph $G$, and not those added to graph $G^*$ by the edge subdivision algorithm.

Let $Z$ denote the integer solution for the \maxmove~problem, which will be defined throughout the proof.
Let $\lambda$ be a parameter chosen uniformly at random from $[\frac{\gamma+1}{3},\gamma]$. For any node $v\in V(G^*)$, if there is an entry $i$ such that $\bar{X}_v^{i}\ge \lambda$, then let $Z_v=\mathbf{e}_i$, see Algorithm \ref{alg:4Mmove}. In other words, if $\bar{X}_v^{i}\ge \lambda$, then node $v$ gets assigned to partition $i$. Observe that since $\gamma> \frac{1}{2}$ then $\lambda\ge \frac{\gamma+1}{3}> 1/2$ and $\sum_{t\in\{1,\ldots,k\}} \bar{X}_v^{t}=1$. So it is not possible for a node $v$ to have more than one entry $i$ for which $\bar{X}_v^{i}\ge \lambda$.

Let $A=\{v\in V(G)|\bar{X}_v^{\ell_v}\ge\lambda\}$ be the set of nodes that get assigned to their initial partitions, i.e., $Z_v=\mathbf{e}_v$. By Lemma \ref{lem:size}, $|V(G)\setminus A|\le \frac{r}{1-\lambda}\le \frac{r}{1-\gamma}$. The above procedure assigns every node $v$ in $A$ to $\ell_v$, so no matter how we continue the rounding we are guaranteed to move no more than a total of $4r$ nodes from their initial partitions. 

Let $S=\{v\in V(G^*)|\exists i\in\{1,\ldots,k\}, \bar{X}_v^{i}\ge\lambda\}$ be the set of all the nodes that are processed by steps 7-10 of Algorithm \ref{alg:4Mmove}, and let $S'=V(G^*)\setminus S$ be the rest of the nodes. 
Random variable $\rho$ is then chosen uniformly at random from interval $[0,\lambda]$. 
Let $B(\rho,i)$ be the set of nodes $v\in S'$ with $\bar{X}_v^i>\rho$. Let $\sigma=(\sigma(1),\sigma(2),\ldots,\sigma(k))$ be one of the two permutations $(1,2,\ldots, k-1,k)$ and $(k-1,\ldots,1,k)$ chosen with equal probabilities at random. Nodes are assigned to partitions in the order of $\sigma$. In particular, for each $j< k$, all the nodes in $B(\rho, \sigma(j))\setminus \cup_{i=1}^{j-1} B(\rho,\sigma(i))$ are placed in partition $\sigma(j)$ and, finally, the remaining nodes are placed in partition $\sigma(k)=k$. Processing each node takes constant time, so the rounding procedure takes linear time.

Let $C$ be the cut-set created by this rounding scheme. Here, we bound the value of $C$. Consider an edge $(u,v)\in E(G^*)$. We provide an upper bound on the probability that edge $(u,v)$ is in cut-set $C$, i.e., $\Pr\big((u,v)\in C\big)$. First, note that if $d_{\bar{X}}(u,v)\ge \frac{2\gamma-1}{3}$, then $\Pr\big((u,v)\in C\big)\le 1\le \frac{3}{2\gamma-1}d_{\bar{X}}(u,v)< \frac{5}{2\gamma-1}d_{\bar{X}}(u,v)$. Therefore, we assume that $d_{\bar{X}}(u,v)< \frac{2\gamma-1}{3}$. Under this assumption, we have three cases for edge $(u,v)$:
\begin{itemize}
\item Case 1: Nodes $u,v\in S$. In this case, we compute $\Pr\big((u,v)\in C|u,v\in S\big)$. Let $i,j$ be the entries for which  $\bar{X}_v^i\ge\lambda$ and $\bar{X}_u^j\ge\lambda$. If $i=j$, then $(u,v)\notin C$. If $i\neq j$, then since $\bar{X}_v^j,\bar{X}_u^i\le 1-\lambda$, we have $d_{\bar{X}}(u,v)\ge \frac{1}{2}\big(|\bar{X}_v^i-\bar{X}_u^i|+|\bar{X}_v^j-\bar{X}_u^j|\big)\ge 2\lambda-1\ge \frac{2\gamma-1}{3}$. This is a contradiction to our assumption of $d_{\bar{X}}(u,v)< \frac{2\gamma-1}{3}$; thus, in this case $(u,v)$ is never in cut-set $C$. 

\item Case 2: Node $u\in S$ and node $v\in S'$. In this case, we compute 
$\Pr\big(u\in S, v\in S'\big)$. By constraint (C4) of LP~\ref{tab:maxmove-lp}, for a node $w$ there can at most be one entry $i$ such that $\bar{X}_w^i\ge \frac{7}{12}$. Let $j$ be such an entry for node $u$. Since $u\in S$ and $v\in S'$, We must have $\bar{X}_u^j\ge \lambda\ge \max(\bar{X}_v^j,\frac{7}{12})$. The probability that such $\lambda$ is chosen is at most $d_{\bar{X}}(u,v)/(\gamma - \frac{\gamma+1}{3}) = \frac{3}{2\gamma-1}d_{\bar{X}}(u,v)$, concluding that $\Pr(u\in S, v\in S') \le \frac{3}{2\gamma-1}d_{\bar{X}}(u,v)$. 


\item Case 3: Nodes $u,v\in S'$. In this case, we compute $\Pr\big((u,v)\in C|u,v\in S'\big)$. This case is very similar to the argument provided in \cite{LP_paper}, which we mention here for the sake of completeness. Suppose that $\bar{X}_u$ and $\bar{X}_v$ only differ in $i$ and $j$ entries. Without loss of generality, suppose that $\bar{X}_u^i=\max(\bar{X}_u^i,\bar{X}_v^i,\bar{X}_u^j,\bar{X}_v^j)$. Subsequently, by constraint (C4) of LP~\ref{tab:maxmove-lp} we have $\bar{X}_u^j=\min(\bar{X}_u^i,\bar{X}_v^i,\bar{X}_u^j,\bar{X}_v^j)$. Let $L_i=(\bar{X}_v^i,\bar{X}_u^i)$ and $L_j=(\bar{X}_u^j,\bar{X}_v^j)$. Note that $|L_i|=|L_j|=d_{\bar{X}}(u,v)$. 
In order for $(u,v)$ to be in cut-set $C$, we must have $\rho\in L_i\cup L_j$. Due to step 13 of Algorithm~\ref{alg:4Mmove}, an entry $1\le i<k$ is processed before an entry $1\le j<k$ with probability 1/2. Suppose that entry $j$ is processed before $i$. If $\rho\in L_j$, then node $v$ goes to partition $j$ and $u$ does not. If $\rho\in L_i\setminus L_j$, then we have that $\bar{X}_v^j,\bar{X}_u^j<\rho$. Therefore, $u$ and $v$ are not assigned to partition $j$, and when entry $i$ is processed, node $u$ is assigned to partition $i$ and $v$ is not. So, when entry $j$ is processed before entry $i$, edge $(u,v)$ is in cut-set $C$ if $\rho\in L_i\cup L_j$. If entry $i$ is processed before $j$ and $\rho\in L_i$, then $u$ is assigned to partition $i$ and $v$ is not. If $\rho\in L_j\setminus L_i$, 
then both $u$ and $v$ are assigned to partition $i$. Therefore, in total, $\Pr\big((u,v)\in C|u,v\in S'\big)\le \frac{1}{2}\Pr(\rho\in L_i\cup L_j)+\frac{1}{2}\Pr(\rho\in L_i)\le \frac{3}{4\lambda}d_{\bar{X}}(u,v)\le \frac{9}{4(\gamma+1)}d_{\bar{X}}(u,v)\le \frac{2}{2\gamma-1}$. 
\end{itemize}
Using these three cases we can quantify the probability that an edge $(u,v)$ lies in cut-set $C$ as follows. Note that we are conditioning on $d_{\bar{X}}(u,v)<\frac{2\gamma-1}{3}$ and so Case 1 does not take place:
\begin{align*}
    \Pr\big((u,v)\in C\big) &= \Pr\big((u,v)\in C|u,v\in S\big)\Pr(u,v\in S)\\
    &+\Pr\big((u,v)\in C|u\in S,v\in S'\big)\Pr(u\in S,v\in S')\\
    &+\Pr\big((u,v)\in C|u,v\in S'\big)\Pr(u,v\in S')\\
    &\le \Pr\big((u,v)\in C|u\in S,v\in S'\big)\frac{3}{2\gamma-1}d_{\bar{X}}(u,v) + \frac{2}{2\gamma-1}d_{\bar{X}}(u,v)\Pr(u,v\in S') \\
    &\le \frac{5}{2\gamma-1}d_{\bar{X}}(u,v).
\end{align*}
\end{proof}

\section{Approximation Algorithm for $k=2$}
\label{subsec:2partitioning}
In this section, we describe our approximation algorithm for the $r$-move $2$-partitioning problem. Since there are two partitions in this problem, there only exists two terminal nodes. We refer to these terminal nodes as $s$ and $t$ to be compliant with the 2-partitioning and the min $s$-$t$ cut problems literature. We first give a high level comparison of our algorithm and that of \cite{zhang2016new}, who studied a related problem called the Min $r$-size $s$-$t$ cut problem and from whom we borrow some ideas for our approximation method. Section~\ref{sec:intro} provides a  formal definition of the Min $r$-size $s$-$t$ cut problem and the distinctions between this problem and the $r$-move $2$-partitioning problem.

\paragraph{Main ideas of \cite{zhang2016new}:} 
Zhang \cite{zhang2016new} uses a problem called the \textit{parametric max flow} problem introduced in \cite{gallo1989fast}. The parametric max flow problem is a generalization of the $s$-$t$ max flow problem where we think of capacities/edge weights as parameters. More formally, the weight of an edge $(u,v)$ can be parameterized by a parameter $\alpha$, such that the weights of the edges connected to terminal $s$, i.e., $c_\alpha(s,v)$, are non-decreasing functions of $\alpha$ and the the weights of the edges connected to terminal $t$, i.e., $c_\alpha(v,t)$, are non-increasing in $\alpha$. Finally, the weights of the edges that are connected to neither of the terminals, i.e., $c_\alpha(u,v)$  for any node $u\neq s,v\neq t$, are constant in terms of $\alpha$. 

Zhang \cite{zhang2016new} creates the following graphs $G^\alpha$ for a parameter $\alpha \ge 0$: Let the weight of all edges $(v,t)$ be $\alpha$, for all nodes $v\neq s$. Clearly this fits into the parametric flow definition of \cite{gallo1989fast}, discussed above. When $\alpha=0$, the min $s$-$t$ cut of $G^\alpha$ is the same as the min $s$-$t$ cut of $G$. Let $S_0$ be the set of nodes in the $s$-side of the min $s$-$t$ cut of $G$. As $\alpha$ gets larger, they show that the $s$-side of the min $s$-$t$ cut in $G^\alpha$ contains fewer and fewer nodes until eventually for some large $\alpha$ the $s$-side contains only $s$. Zhang \cite{zhang2016new} argues that by considering values of $\alpha$ from $0$ to some large number, there are $h\le n-1$ distinct sets $S_0,\ldots, S_h=\{s\}$ corresponding to the $s$-sides of the min $s$-$t$ cuts of $G^{\alpha}$ graphs, such that $|S_i|> |S_{i+1}|$ for all $i=0,\ldots,h-1$. Gallo et al. \cite{gallo1989fast} show that one can compute the sets $S_i$ in $O(mn\log(n^2/m))$ time. Zhang \cite{zhang2016new} proceeds by proving a few properties for $G^\alpha$ graphs and using the linear program formulation of the \stcutatmost~problem they give a $\frac{k+1}{k+1-k^*}$-approximation algorithm for it. 

We can approach the \maxmove[2] problem using similar $G^\alpha$ graphs. More specifically, a graph $G^\alpha$ is produced by adding edges $(s,v)$ and $(v,t)$ with weights $\alpha$ to the original graph $G$, for each $v\neq s,t$. These graphs do not fit into the definition of parametric max flow problem as $c_{\alpha}(s,v)=\alpha$ and $c_{\alpha}(v,t)=\alpha$ both grow in the same direction and we cannot use the results of \cite{gallo1989fast}. Nevertheless, we are able to produce the same results as~\cite{zhang2016new} for our setting, and essentially show that the approach of Zhang can be generalized to work for the \maxmove[2] problem. Assuming that $r^*\le r$ is the optimal number of nodes moved in \maxmove[2], we first give a simple $(r^*+1)$-approximation algorithm for the \maxmove[2]~problem without using linear programming. In Subsection~\ref{sec:completecase}, we include a more sophisticated $\frac{r+1}{r+1-r^*}$-approximation algorithm for the \maxmove[2]~problem. 

\subsection{An $(r^*+1)$-Approximation Algorithm}
Let $G$ be an instance of the \maxmove[2]~problem. We first define additional graphs constructed based on $G$ and continue by stateing a few observations about them. Suppose that in the initial partitioning, terminal $s$ is in cluster $1$ and terminal $t$ is in cluster $2$. For any $\alpha\ge 0$, we construct graph $G^\alpha$ as follows: Take graph $G$ and add edges from terminal $s$ to any node $v$ in cluster $1$ with weight $\alpha$. If there already exists an $(s,v)$ edge in $G$, then simply add $\alpha$ to the weight of this edge. Similarly, add edges from any node $v$ in cluster $2$ to terminal $t$ with weight $\alpha$, and if the edge $(v,t)$ exists in $G$, then add $\alpha$ to its weight\footnote{Note that the $\alpha_i$ defined here is different from the $\alpha_i$ used in \cite{zhang2016new}.}. 

Let the value of the initial cut of $G$ be $c_\infty$. Note that if we set $\alpha\ge 2c_\infty$, then any cut other than the initial cut will have value at least $2c_\infty$, so the min $s$-$t$ cut for $G^{2c_\infty}$ is the initial cut. If we let $\alpha=0$, then the min $s$-$t$ cut for $G^0$ is the min $s$-$t$ cut for $G$. 

Any cut in graph $G$ 
 results in a set of nodes $S\subseteq V(G)\setminus\{s,t\}$ being moved from their initial partitions. We say that a cut is \emph{induced, formed} or \emph{produced} by set $S\subseteq V(G)\setminus\{s,t\}$ if this cut leads to the set of nodes $S$ being moved to partitions different from their initial ones. For example, when $S=\emptyset$, the cut formed by $S$ is the initial cut of the graph $G$ and when $S=\{v\}$, the cut induced by $S$ moves node $v$, and only node $v$, from its initial partition to the other one. We denote the size of a set $S$ by $r$, that is, $|S|=r$. 

For any set $S\subseteq V(G)\setminus\{s,t\}$, let $\delta(S)$ be the value of the cut formed by $S$ in the original graph $G$. Then, if the min $s$-$t$ cut in $G^\alpha$ moves nodes in set $S_\alpha$, the value of this cut in $G^\alpha$ is $c_\alpha = r_\alpha\alpha+\delta(S_\alpha)$, where $r_\alpha = |S_\alpha|$. Note that there might be several sets that produce a min $s$-$t$ cut for $G^\alpha$ and we consider one such set arbitrarily as $S_\alpha$. Observe that if $S$ and $S'$ are both min $s$-$t$ cuts in $G^\alpha$ such that $|S'|=|S|$, then $\delta(S)=\delta(S')$, i.e., they have the same cut value in $G$. The following lemma will help us in developing our approximation algorithm for the problem.

\begin{lemma}\label{lem:main_properties}
Suppose the min $s$-$t$ cuts in graphs $G_\alpha$ and $G_{\alpha'}$ move sets of nodes $S_\alpha$ and $S_{\alpha'}$ from their initial partitions, respectively. 
We have $\alpha\ge \alpha'$ if and only if $|S_\alpha|\le |S_{\alpha'}|$. Moreover, if we have $|S_\alpha|< |S_{\alpha'}|$ then it implies that $\delta(S_\alpha)> \delta(S_{\alpha'})$.
\end{lemma}
\begin{proof}
Since $S_\alpha$ and $S_{\alpha'}$ produce min $s$-$t$ cuts for $G^\alpha$ and $G^{\alpha'}$, respectively, we have the following two inequalities:
$$
\alpha |S_\alpha|+\delta(S_{\alpha})\le \alpha |S_{\alpha'}|+\delta(S_{\alpha'})
$$
$$
\alpha'|S_{\alpha'}|+\delta_(S_{\alpha'}) \le \alpha'|S_\alpha|+\delta(S_{\alpha}).
$$
By reordering and combining these inequalities we get 
\begin{equation}\label{eq:helper1}
\alpha'(|S_{\alpha'}|-|S_\alpha|)\le \delta(S_{\alpha})-\delta(S_{\alpha'})\le \alpha(|S_{\alpha'}|-|S_\alpha|).\end{equation}
From these two inequalities we have $(\alpha-\alpha')(|S_{\alpha'}|-|S_\alpha)|)\ge 0$. Thus, $\alpha-\alpha'$ and $|S_{\alpha'}|-|S_\alpha|$ have the same signs and this proves the first part of the lemma. Moreover, by inequalities \eqref{eq:helper1} we have$|S_\alpha|< |S_{\alpha'}|$ implies that $\delta(S_\alpha)> \delta(S_{\alpha'})$.
\end{proof}
\begin{corollary}\label{cor:breakpoints}
Let $\alpha<\alpha'$. If $S$ induces a min $s$-$t$ cut for both $G^{\alpha}$ and $G^{\alpha'}$, then for any $\hat{\alpha}\in [\alpha,\alpha']$ set $S$ forms a min $s$-$t$ cut for graph $G^{\hat{\alpha}}$.
\end{corollary}
\begin{proof}
Suppose $\hat{S}$ forms a min $s$-$t$ cut for graph $G^{\hat{\alpha}}$. Set $S$ produces a min $s$-$t$ cut for both $G^{\alpha}$ and $G^{\alpha'}$; therefore, $S_\alpha = S_{\alpha'} = S$ and $\delta(S_\alpha) = \delta(S_{\alpha'}) = \delta(S)$. Since $\hat{\alpha}>\alpha'$, by Lemma~\ref{lem:main_properties} we have $|S_{\hat{\alpha}}|\le |S_{\alpha'}|=|S|$ and $\delta(S) = \delta(S_{\alpha'})\le \delta(S_{\hat{\alpha}})$. Similarly, since $\hat{\alpha}<\alpha$, we have $|S_{\hat{\alpha}}|\ge |S_\alpha|=|S|$ and $\delta(S)=\delta(S_\alpha) \ge \delta(S_{\hat{\alpha}})$. Subsequently, $|S_{\hat{\alpha}}|= |S|$ and $\delta(S)= \delta(S_{\hat{\alpha}})$. Therefore, $S$ induces a min $s$-$t$ cut for graph $G^{\hat{\alpha}}$.
\end{proof}

Lemma~\ref{lem:main_properties} states that as $\alpha$ gets larger, the set of nodes that are moved from their initial partitions gets smaller and the value of the associated cut in $G$ gets larger. Using this intuition, we define \emph{breakpoints} as the different number of nodes moved to partitions different from their initial ones by min $s$-$t$ cuts on graphs $G^0,\ldots,G^{2c_\infty}$. Definition~\ref{def:breakpoints}  formally remarks the definition of breakpoints.

\begin{definition}[Breakpoints]\label{def:breakpoints}
For some $0\le h\le n-1$, a set $\{r_0,\ldots,r_h\}$ is called a set of breakpoints with each $r_i$ being one breakpoint if $r_0>r_1>\ldots>r_h=0$ and there are sets $S_0, S_1,\ldots,S_h$ such that $|S_i|=r_i$. Furthermore, for any $\alpha\ge 0$, there is an index $0\le i\le h$ such that $S_i$ is a min $s$-$t$ cut of $G^{\alpha}$, and for any $i$, there exists at least an $\alpha\ge 0$ such that $S_i$ is a min $s$-$t$ cut of graph $G^{\alpha}$.
\end{definition}


Note that for any graph $G$ there might be more than one set of breakpoints. 
%
 %
%
We show that by running $O(h)\in O(n)$ max flow algorithms we can compute a set of  breakpoints and their associated sets, i.e., sets of $r_i$'s and $S_i$'s for all $0\le i\le h$, respectively. This algorithm is demonstrated in Algorithm \ref{alg:findS}. 


\begin{algorithm}[h]
\caption{Computing breakpoints}\label{alg:findS}
\begin{algorithmic}[1]
    \State $c_\infty\gets$ the value of the initial cut of $G$. 
    \State Run the min $s$-$t$ cut algorithm on $G$ to get a set of nodes $S_0$ that are moved from their initial partitioning.
    \If{$S_0=\emptyset$}
        \State \Return $\{0\},\{\emptyset\}$.
    \EndIf
    \State $R \gets \{0,|S_0|\}$ .
    \Comment{The set of breakpoints.}
    \State $R' \gets \{\emptyset,S_0\}$.
    \Comment{The set of sets associated to the breakpoints.}
    \State $I \gets \{(\emptyset,S_0)\}$. \Comment{The set of active pairs.}
    \While{$I\neq \emptyset$}
        \State $(S',S)\gets$ an active pair from $I$. Let $r=|S|$ and $r'=|S'|$.
        \State $\bar{\alpha} \gets \frac{\delta(S')-\delta(S)}{r-r'}$. Run min $s$-$t$ cut algorithm on $G^{\bar{\alpha}}$. Let the value of the cut be $c_{\bar{\alpha}}$. 
        \If{$c_{\bar{\alpha}}=r\bar{\alpha}+\delta(S)$ 
        } \Comment{Equivalently, $r\bar{\alpha}+\delta(S) = r'\bar{\alpha}+\delta(S')$.}
            \State Remove $(S',S)$ from $I$.
        \Else   
            \State Let the min $s$-$t$ cut in $G^{\bar{\alpha}}$ be produced by set $S''$. 
            \State Add $S''$ to $R'$ and $|S''|$ to $R$.  
            \State Remove $(S',S)$ from $I$ and add $(S',S'')$ and $(S'',S)$ to $I$.
        \EndIf
    \EndWhile
    \State \Return $R, R'$.  
\end{algorithmic}
\end{algorithm}
 
\begin{lemma}\label{lem:findS}
Algorithm \ref{alg:findS} finds the breakpoints for an instance $G$ of the \maxmove[2]~problem in $O(n)T_{min-cut}(n)$, where $T_{min-cut}(n)$ is the running time of a  min $s$-$t$ cut algorithm on an $n$-node graph. 
\end{lemma}
\begin{proof}
First, observe that if the initial partitioning in $G$ is a min $s$-$t$ cut of the graph, then $S_0=\emptyset$ and for any $\alpha\ge 0$ set $S_0$ induces a min $s$-$t$ cut of $G^\alpha$. Note that if the initial partitioning in $G$ is a min $s$-$t$ cut, then $S_0$ induces the min $s$-$t$ cut for $G^0$ and $G^{2c_\infty}$. By Corollary \ref{cor:breakpoints}, set $S_0$ induces a min $s$-$t$ cut for any graph $G^\alpha$ for $\alpha\in [0,2c_\infty]$. Moreover, for any $\alpha\ge 2c_\infty$ it is easy to see that that $\emptyset$ induces a min $s$-$t$ cut for $G^\alpha$. Therefore, in this case there exists only one breakpoint, that is $0$. In the rest of this proof, we assume that $S_0\neq \emptyset$.


The algorithm starts with an incomplete set of breakpoints $R=\{0,|S_0|\}$, the set of sets associated with these breakpoints $R'=\{\emptyset,S_0\}$ and the set of active pairs $I=\{(\emptyset,S_0)\}$. At each iteration of the algorithm, we start with a pair of sets of nodes $(S',S)$ that we call an ``active pair" such that $S',S\in R$ and $|S'|<|S|$. This active pair is either broken into two active pairs and a new breakpoint is added to $R$ by the algorithm, or the active pair is removed from the set. We show that the set of all values $\alpha$ such that $S$ or $S'$ induce a min $s$-$t$ cut for $G^\alpha$ is a contiguous interval. Recall that  $S_0$ forms a min $s$-$t$ cut of $G^0$ and for any $\alpha\ge 2c_\infty$, we have that $\emptyset$ induces a min $s$-$t$ cut of $G^\alpha$. Since we start off with an active pair $(\emptyset,S_0)$, for any $\alpha\ge 0$ the algorithm returns a set $S_\alpha$ such that $S_\alpha$ forms a min $s$-$t$ cut of $G^\alpha$. 
Whenever the algorithm adds a new breakpoint to set $R$, the number of active pairs in $I$ increases by one. Note that 
the number of breakpoints cannot be larger than $n$, which means that the total number of active pairs that are processed in the algorithm cannot be larger than $n+1$. So, the running time of the algorithm is $O(n)T_{min-cut}(n)$.



Consider an active pair $(S',S)$ and let $r=|S|$ and $r'=|S'|$ such that $r'<r$. Suppose that there exists an $\alpha$ and an $\alpha'$ such that $S$ and $S'$ form min $s$-$t$ cuts in graphs $G^{\alpha}$ and $G^{\alpha'}$, respectively. Since $r'<r$, by Lemma~\ref{lem:main_properties} we have we have $\alpha'>\alpha$. Let $\bar{\alpha} = ({\delta(S')-\delta(S)})/({r-r'})$. Since set $S'$ forms a min $s$-$t$ cut in graph $G^{\alpha'}$, we have $\alpha'r'+\delta(S')< \alpha'r+\delta(S)$ which infers $\bar{\alpha}< \alpha'$. With similar calculations, it can be shown that $\alpha < \bar{\alpha}$; thus, $\alpha < \bar{\alpha} < \alpha'$. We then run the min $s$-$t$ cut algorithm on $G^{\bar{\alpha}}$. 
If the value of a min $s$-$t$ cut in $G^{\bar{\alpha}}$ is smaller than $r\bar{\alpha}+\delta(S)$,  
then let $S''$ be the set forming the min $s$-$t$ cut in $G^{\bar{\alpha}}$. 
Since $\alpha < \bar{\alpha} < \alpha'$, using Lemma \ref{lem:main_properties}, we have that $r'<|S''|<r$.
So we add $|S''|$ as another breakpoint to $R$ and its associated set $S''$ to $R'$. We also need to do an iteration of the algorithm on $(S',S'')$ and another iteration on $(S'',S)$. To do that, we add these two pairs to the set of our active pairs $I$ and remove $(S',S)$ from it.

Now, consider the case that the value of a min $s$-$t$ cut in $G^{\bar{\alpha}}$ is $c_{\bar{\alpha}}=r\bar{\alpha}+\delta(S)$. From the definition of $\bar{\alpha}$, we also have $c_{\bar{\alpha}}=r'\bar{\alpha}+\delta(S')$. Let $\alpha$ and $\alpha'$ be two values such that $S$ and $S'$ induce min $s$-$t$ cuts of $G^\alpha$ and $G^{\alpha'}$, respectively. First, since $r'<r$, by Lemma \ref{lem:main_properties}, we have that $\alpha'\ge\alpha$. 
We argue that in this case there is no breakpoint between $r$ and $r'$, in other words, for any $\alpha''\in [\alpha,\alpha']$ either $S$ or $S'$ induces a min $s$-$t$ cut of $G^{\alpha''}$. More precisely, since $S$ induces a min $s$-$t$ cut for both $G^{\alpha}$ and $G^{\bar{\alpha}}$, if $\alpha''\in [\alpha,\bar{\alpha}]$ then  by Corollary \ref{cor:breakpoints} $S$ induces a min $s$-$t$ cut for $G^{\alpha''}$, as well. Similarly, since $S'$ induces a min $s$-$t$ cut for both $G^{\alpha'}$ and $G^{\bar{\alpha}}$, if $\alpha''\in [\bar{\alpha},\alpha']$, then by Corollary \ref{cor:breakpoints} $S'$ induces a min $s$-$t$ cut for $G^{\alpha''}$, too. Therefore, there is no need to search for additional breakpoints between $r$ and $r'$ and we can safely remove $(S',S)$ from the set of active pairs. 
%
%

We show that when the algorithm terminates, for any $\alpha\ge 0$ there is a set $S_\alpha\in R'$ such that $S_\alpha$ produces a min $s$-$t$ cut of $G^\alpha$.
If $\alpha\ge 2c_\infty$, then $S=\emptyset$ has this property. If $\alpha=0$, then $S=S_0$ satisfies this condition. It only remains to show that this claim holds for any $0< \alpha< 2c_\infty$. When the algorithm terminates, let $r_0>\ldots>r_h=0$  be the breakpoints in $R$ with corresponding sets $S_0,\ldots,S_h$, where $r_i=|S_i|$. For each $0\le i\le h$, let $\alpha_i'$ be the value for which $S_i$ forms a min $s$-$t$ cut for $G^{\alpha_i'}$, and let $\alpha_0=0$ and $\alpha_h=2c_\infty$. There is an index $0\le i<h$ such that $\alpha_i'\le \alpha\le \alpha_{i+1}'$. Since there is no breakpoint between $r_i$ and $r_{i+1}$, there is a stage of the algorithm where the active pair $(S_{i+1},S_{i})$ gets processed and then gets removed from $I$. So from the argument in the previous paragraph, either $S_i$ or $S_{i+1}$ produces a min $s$-$t$ cut for $G^\alpha$. Finally, note that if the algorithm adds a breakpoint $r''$ with associated set $S''$ to $R$, then $S''$ induces a min $s$-$t$ cut of $G^{\bar{\alpha}}$ for some $\bar{\alpha}\ge 0$, concluding that set $R$ contains all required breakpoints and every sets associated with the breakpoints are included in $R'$.
%
%
%
\end{proof}

Before proceeding to our algorithm, we prove an important property of the breakpoints.
\begin{lemma}\label{lem:alphai}
   Let $R=\{r_0,\ldots,r_h\}$ be a set of 
   breakpoints with associated sets $S_0,\ldots,S_h$. For any $0\le i<h$, let $\alpha_i=(\delta(S_i)-\delta(S_{i+1}))/(r_{i+1}-r_{i})$. We have that both $S_i$ and $S_{i+1}$ induce min $s$-$t$ cuts for $G^{\alpha_i}$.
\end{lemma}
\begin{proof}
Since $R$ is the set of breakpoints, there exist $\alpha_i'$ and $\alpha_{i+1}'$ such that $S_i$ and $S_{i+1}$ induce min $s$-$t$ cuts for $G^{\alpha_i'}$ and $G^{\alpha_{i+1}'}$, respectively. Therefore, we have $\alpha_i'r_{i}+\delta(S_i)\le \alpha_i'r_{i+1}+\delta(S_{i+1})$ and  $\alpha_{i+1}'r_{i+1}+\delta(S_{i+1})\le \alpha_{i+1}'r_{i}+\delta(S_{i})$, which infer that $\alpha_i'\le \alpha_i\le \alpha_{i+1}'$. Moreover, since $R$ is a set of breakpoints and sets $S_0,\ldots,S_h$ correspond to the breakpoint, there exists an index $0\le j\le h$ such that $S_j$ forms a min $s$-$t$ cut for $G^{\alpha_i}$. 
By Lemma \ref{lem:main_properties}, from $\alpha_i'\le \alpha_i\le \alpha_{i+1}'$ we have that $r_i\ge r_j=|S_j|\ge r_{i+1}$; therefore, $j\in \{i,i+1\}$. Note that from the definition of $\alpha_i$, we have $\alpha_ir_i+\delta(S_i)=\alpha_ir_{i+1}+\delta(S_{i+1})$. Putting this together with the fact that the value of the min $s$-$t$ cut in $G^{\alpha_i}$ is $\alpha_ir_j+\delta(S_j)$, we have that both $S_i$ and $S_{i+1}$ induce min $s$-$t$ cuts for $G^{\alpha_i}$. 
\end{proof}

Our algorithm for the \maxmove[2]~problem works as follows: First, using Algorithm~\ref{alg:findS}, we compute the breakpoints $r_0>\ldots>r_h=0$ and their associated sets $S_0,\ldots,S_h$ such that $|S_i|=r_i$ for all $0\le i\le h$. 
Let $j$ be the largest index such that $r_j\le r$. The $s$-$t$ cut produced by $S_j$ is returned as the approximate \maxmove[2]. Our algorithm is represented in Algorithm \ref{alg:2part}. Note that, even though our algorithm looks similar to that of \cite{zhang2016new}, its first step is computed differently. To analyze the algorithm, we first state a few lemmas that also appear in \cite{zhang2016new} and we show that they hold in our setting, as well. 

\begin{algorithm}[h]
\caption{Approximation algorithm for the \maxmove[2]~problem}\label{alg:2part}
\begin{algorithmic}[1]
    \State Using Algorithm \ref{alg:findS}, compute breakpoints $r_0>\ldots>r_h=0$ and their associated sets $S_0,\ldots,S_h$, such that $|S_i|=r_i$.
    \If{$r\ge r_0$}
    \State \Return $S_0$.
    \Else
    \State $j\gets$ the index for which $r_{j-1}>r\ge r_{j}$  
    \State \Return $S_{j}$.
    \EndIf
\end{algorithmic}
\end{algorithm}

\begin{lemma}\label{lem:citation_lemmas}
For graph $G$, let $S_0$ denote the set of nodes moved by a min $s$-$t$ cut to partitions different from their initials ones and let $r_0=|S_0|$. Moreover, suppose that for some $0\le h\le n-1$, $\{r_0,\ldots,r_h\}$ is a set of breakpoints for $G$, with their associated sets being $S_0, S_1,\ldots,S_h$, such that $|S_i|=r_i$. Let $r^*$ denote the number of nodes moved from their initial partitions by an optimal $r$-move 2-partitioning. The following properties hold for an optimal $r$-move 2-partitioning in $G$:
\begin{enumerate}
    \item If $r\ge r_0$, then $S_0$ provides an optimal $r$-move 2-partitioning in $G$.
    \item If $r<r_0$ and there exists some $0<i\le h$ such that $r_i=r^*$, then the largest index $j$ for which $r_j\le r$ is $i$ and the corresponding set $S_i$ induces an optimal $r$-move 2-partitioning in $G$.
    \item If $r<r_0$ and $r^*\notin \{r_0,\ldots,r_h\}$, then none of the values of $r,r-1,\ldots,r^*-1$ are among the breakpoints, i.e., $r_i$'s.
    \item If $r=r_i$ for some $0\le i\le h$, then set $S_i$ produces an optimal $r$-move 2-partitioning in graph $G$. 
\end{enumerate}
\end{lemma}
\begin{proof}
We denote the set that induces the optimal \maxmove[2]~for graph $G$ by $S^*$, where $|S^*|=r^*$. In the first property, it is assumed that $r\ge r_0$; therefore, set $S_0$ moves at most $r$ nodes from their initial partitions. Thus, $S_0$ is an optimal $r$-move 2-partitioning in $G$.

To prove the second property, for the sake of contradiction,
suppose that there exists an index $j$ such that $i < j\le h$ and $r^*=r_i<r_j\le r$.
Let $S_i$ and $S_j$ be the associated sets and let $\alpha$ and $\alpha'$ be the values such that $S_i$ and $S_j$ induce min $s$-$t$ cuts of $G^\alpha$ and $G^{\alpha'}$, respectively. In other words, we have $S_i=S_\alpha$, $S_j= S_{\alpha'}$, $|S_i|=r_i$ and $|S_j|=r_j$. First, from the optimality of $S_i$ for $G^\alpha$ we have that $\alpha r_i+\delta(S_i)\le \alpha r^*+\delta(S^*)$. Since $r_i=r^*$, we have $\delta(S_i)\le \delta(S^*)$. Due to the optimality of the cut induced by set $S^*$ among all $s$-$t$ cuts that move $r^*$ nodes to partitions different from their initial ones, we have $\delta(S_i)=\delta(S^*)$. Now since $r_i<r_j$, by Lemma \ref{lem:main_properties} we have that $\delta(S_i) = \delta(S_\alpha)>\delta(S_{\alpha'}) = \delta(S_j)$ (Note that $\delta(S_\alpha)$ is strictly greater than $\delta(S_{\alpha'})$ as $r_i$ is strictly less than $r_j$). Therefore, set $S_j$ forms a smaller cut for $G$ by moving at most $r$ nodes, and this contradicts the optimality of the cut induced by set $S^*$ and the second property holds.

To prove the third property, again for the sake of contradiction, suppose that there exists an $r'$ such that $r\ge r'> r^*$ and $r'=r_i$ for some $0\le i\le h$. Let $S_i$ be the set associated to $r_i$ and let $\alpha$ be a value for which $S_i$ forms a min $s$-$t$ cut for graph $G^\alpha$. We have $\alpha r_i+\delta(S_i)\le \alpha r^*+\delta(S^*)$. Since $r'=r_i > r^*$, we have $\delta(S_i)<\delta(S^*)$. Since $r_i\le r$, having $\delta(S_i)<\delta(S^*)$ contradicts the optimality of the $s$-$t$ cut induced by set $S^*$ and the third property holds.

For the fourth property, note that if $r=r_0$, then by the first property set $S_0$ forms an optimal \maxmove[2] for $G$. If $r=r^*=r_i$, then the second property gives us the result. If $r_i=r>r^*$, then this contradicts the third property, so the only way that $r=r_i$ is that $r=r^*$ or $r=r_0$.
\end{proof}


\begin{theorem}\label{thn:2part-weaker}
There is an $(r^*+1)$-approximation algorithm that solves the weighted instances of the \maxmove[2]~problem, where $r^*$ is the number of nodes moved from their initial partitions by an optimal solution. 
This approximation algorithm runs in time $O(n)T_{min-cut}(n)$, where $T_{min-cut}(n)$ is the running time of a $s$-$t$ min-cut algorithm on an $n$-node graph.
\end{theorem}
\begin{proof}
We show that Algorithm \ref{alg:2part} provides the claimed guarantees. Let $G$ be the original graph of $n$-nodes. Let  $0\le h\le n-1$ and $\{r_0,\ldots,r_h\}$ be the set of breakpoints computed in step 1 of the algorithm, with their associated sets being $S_0, S_1,\ldots,S_h$, such that $|S_i|=r_i$. Let $S^*$ denote the set that induces the optimal \maxmove[2]~for graph $G$ with a cut value of $\delta(S^*)$, where $|S^*|=r^*$. 

If $r\le r_0$, then from the first property of Lemma \ref{lem:citation_lemmas} we have that $S_0$ provides an optimal $r$-move $2$-partitioning of $G$ which is the set that the algorithm outputs. So suppose that $r>r_0$.

From the second and fourth properties of Lemma \ref{lem:citation_lemmas}, 
we have that if $r^*$ or $r$ are among the breakpoints, then the algorithm finds the set inducing the optimal cut. Otherwise, if $r_{j}$ is the largest breakpoint smaller than $r$, then we have $r_{j}<r^*\le r<r_{j-1}$. This is because if $r_{j+1}<r^*< r_j\le r<r_{j-1}$, then by Lemma~\ref{lem:main_properties} set $S_j$ forms a smaller cut than $S^*$ in $G$, contradicting the optimality of the $s$-$t$ cut induced by $S^*$. The rest of this proof focuses on the case that  $r_{j}<r^*\le r<r_{j-1}$.

Let $\alpha=\alpha_{j-1}$ as defined in Lemma \ref{lem:alphai}, such that for $r_{j-1} = |S_{j-1}|, r_{j}=|S_{j}|$ we have $\alpha r_{j-1}+\delta(S_{j-1})=\alpha r_{j} +\delta(S_{j})$. 
Since $S_{j-1}$ and $S_{j}$ produce min $s$-$t$ cuts for $G^\alpha$, we have that $\alpha r^*+\delta(S^*)\ge \alpha r_{j-1}+\delta(S_{j-1})=\alpha r_{j} +\delta(S_{j})$, providing the following two inequalities
\begin{align}
    \delta(S^*)-\delta(S_{j-1})\ge \alpha(r_{j-1}-r^*), \label{eq:2apx-1}\\
    \delta(S_{j})\le \delta(S^*)+\alpha(r^*-r_{j}). \label{eq:2apx-2}
\end{align}
Since $r_{j-1}>r^*$, from inequality~\eqref{eq:2apx-1} we have $\delta(S^*)-\delta(S_{j-1})\ge \alpha(r_{j-1}-r^*)\ge \alpha$, providing $\alpha \le \delta(S^*)$. Combining this with inequality~\eqref{eq:2apx-2}, we have $\delta(S_{j})\le \delta(S^*)(1+r^*-r{j})\le \delta(S^*)(r^*+1)\le \delta(S^*)(r+1)$. Therefore, the output of Algorithm~\ref{alg:2part}, i.e., $S_j$, produces an $(r^*+1)$-approximation solution for the \maxmove[2]~problem. By Lemma \ref{lem:findS}, the running time of step 1 of Algorithm~\ref{alg:2part} is $O(n)T_{min-cut}(n)$. The rest of the algorithm works in linear time, providing an overall running time of $O(n)T_{min-cut}(n)$.
\end{proof}

\subsection{$\frac{r+1}{r+1-r^*}$-Approximation}\label{sec:completecase}

We show that Algorithm \ref{alg:2part} has an approximation factor of $\frac{r+1}{r+1-r^*}$.

\begin{restatable}{theorem}{thmtwopart}\label{thm:2part}Given a positive integer $r$, there is an approximation algorithm that solves the weighted $n$-node $m$-edge instances of the \maxmove[2] problem with an approximation factor $\frac{r+1}{r+1-r^*}$, where $r^*$ is the number of nodes moved from their initial partitions by an optimal solution. 
The algorithm runs in  $O(n)T_{min\text{-}cut}(m)$ time, where $T_{min\text{-}cut}(m)$ is the running time of the $s$-$t$ min-cut algorithm on an $m$-edge graph.
\end{restatable}

Note that when $k=2$, in LP~\ref{tab:maxmove-lp} we can have one variable $x_v$ instead of a vector $X_v$, and if the terminals are represented by $s$ and $t$, then we let $x_v=1$ stand for the case that node $v$ is assigned to the same partition as terminal $s$, and $x_v=0$ otherwise (i.e., the case that node $v$ is allocated the same partition as terminal $t$). By defining $y_{uv} = |x_u-x_v|$ we have $d(u,v)=y_{uv}$. Let cluster $1$ be the cluster that terminal $s$ is in and cluster $2$ be the cluster that terminal $t$ is in. Using these notation we can simplify LP~\ref{tab:maxmove-lp} to LP~\ref{tab:maxmove2-lp}, see Table~\ref{tab:maxmove2-lp}. Note that if node $v$ is initially in cluster $1$, that is, $\ell_v=1$, then $(1-x_v)$ is $1$ if $v$ has moved to cluster $2$ and it is zero otherwise. If node $v$ is initially in cluster $2$, that is, $\ell_v=2$, then $x_v$ is $1$ if $v$ has moved to cluster $1$ and zero otherwise. 

\begin{table}[hht]
    \centering
    \begin{tabular}{lllr}
    \toprule
        \text{Minimize:} & $\sum_{(u,v)\in E(G)}c_{uv}y_{uv}$ & &\\
        \text{Subject to:} & & &\\
        & $y_{uv}\ge x_u-x_v$ & $\forall (u,v)\in E(G)$ & (C1)\\
        & $y_{uv}\ge x_v-x_u$ & $\forall (u,v)\in E(G)$ & (C2)\\
        & $x_{s}=1$ &  & (C3)\\
        & $x_{t}=0$ &  & (C4)\\
        & $0\le x_v$ & $\forall v\in V(G)$ & (C5)\\
        & $\sum_{v\in V(G), \ell_v=1}(1-x_{v})+\sum_{v \in V(G), \ell_v=2}x_{v} \le r$ & & (C6)\\
        \bottomrule
    \end{tabular}
    \caption{$r$-move $2$-partitioning LP (LP~\ref{tab:maxmove2-lp})}
    \label{tab:maxmove2-lp}
\end{table}

Next, using a non-negative parameter $\alpha'$ we formulate the Lagrangian relaxation of LP~\ref{tab:maxmove2-lp} to obtain LP~\ref{tab:lagrangian}, see Table~\ref{tab:lagrangian}. We can show that LP~\ref{tab:lagrangian} has integer optimal solutions by showing that it is in fact the min $s$-$t$ cut LP for graph $G^{\alpha'}$. Recall that as discussed in Section~\ref{subsec:2partitioning}, to obtain a graph $G^\alpha$ we add $\alpha$ weighted edges $(s,v)$ and $(t,u)$ for each $v$ and $u$ with $\ell_v=1$ and $\ell_u=2$. If for each $v$ and $u$ we let $y_{sv}=1-x_v = |x_s-x_v|$ and $y_{tu} = x_u = |x_t-x_u|$, respectively, then we can rewrite  LP~\ref{tab:lagrangian} as LP~\ref{tab:stcut}, see Table~\ref{tab:stcut}. Since LP~\ref{tab:stcut} has integer optimal solutions \cite{wong1983combinatorial,zhang2016new}, LP~\ref{tab:lagrangian} also has integer optimal solutions.  

The rest of the proof follows similar to \cite{zhang2016new}, which we state here for completeness. Recall that $S^*$ is an optimal solution to the \maxmove[2]~problem with $|S^*|=r^*$ and $\delta(S^*)$ is the value of cut induced by $S^*$ in the original graph $G$.

\begin{table}[hht]
    \centering
    \begin{tabular}{lllr}
    \toprule
        \text{Minimize:} & $\alpha'(\sum_{v\in V(G), \ell_v=1}(1-x_{v})+\sum_{v\in V(G), \ell_v=2}x_{v}) + \sum_{(u,v)\in E(G)}c_{uv}y_{uv}$ & &\\
        \text{Subject to:} & & &\\
        & $y_{uv}\ge x_u-x_v$ & $\forall (u,v)\in E(G)$ & (C1)\\
        & $y_{uv}\ge x_v-x_u$ & $\forall (u,v)\in E(G)$ & (C2)\\
        & $x_{s}=1$ &  & (C3)\\
        & $x_{t}=0$ &  & (C4)\\
        & $0\le x_v$ & $\forall v\in V(G)$ & (C5)\\
        \bottomrule
    \end{tabular}
    \caption{Lagrangian of LP~\ref{tab:maxmove2-lp} (LP~\ref{tab:lagrangian})}
    \label{tab:lagrangian}
\end{table}

\begin{table}[hht]
    \centering
    \begin{tabular}{lllr}
    \toprule
        \text{Minimize:} & $ \sum_{(u,v)\in E(G^{\alpha'})}c_{uv}y_{uv}$ & &\\
        \text{Subject to:} & & &\\
        & $y_{uv}\ge x_u-x_v$ & $\forall (u,v)\in E(G^{\alpha'})$ & (C1)\\
        & $y_{uv}\ge x_v-x_u$ & $\forall (u,v)\in E(G^{\alpha'})$ & (C2)\\
        & $x_{s}=1$ &  & (C3)\\
        & $x_{t}=0$ &  & (C4)\\
        & $0\le x_v$ & $\forall v\in V(G^{\alpha'})$ & (C5)\\
        \bottomrule
    \end{tabular}
    \caption{Min $s$-$t$ cut  LP for $G^{\alpha'}$ (LP~\ref{tab:stcut})}
    \label{tab:stcut}
\end{table}


\begin{lemma}\label{lem:helper}
Let $S_j$ be the set of nodes that Algorithm \ref{alg:2part} outputs. If for some $\lambda>0$ we have $|S_{j-1}|\ge\frac{1}{1-\lambda}r^*$, then $\delta(S_{j})\le \frac{1}{\lambda}\delta(S^*)$.
\end{lemma}
\begin{proof}
By Lemma~\ref{lem:alphai} we have that for $\alpha':=\alpha_{j-1}=\frac{\delta(S_{j})-\delta(S_{j-1})}{r_{j-1}-r_{j}}$, both $S_j$ and $S_{j-1}$ produce min $s$-$t$ cuts for graph $G^{\alpha'}$. LP~\ref{tab:lagrangian} and LP~\ref{tab:stcut} provide integer optimal solutions to the min $s$-$t$ cut problem on graph $G^{\alpha'}$. Therefore, the $s$-$t$ cuts produced by sets $S_j$ and $S_{j-1}$ are both optimal solutions of LP~\ref{tab:lagrangian} and LP~\ref{tab:stcut} for parameter $\alpha'$. 
Let the LP~\ref{tab:lagrangian} integer solutions associated with sets $S_{j-1}$ and $S_{j}$ be $(x^-,y^-)$ and $(x^+,y^+)$, respectively. 
Any linear combination of these two solutions is also an optimal solution of LP~\ref{tab:lagrangian}. We choose $\gamma\in (0,1)$ such that 
\begin{equation}\label{eq:gamma}
(1-\gamma)(\sum_{v, \ell_v=1}(1-x^-_{v})+\sum_{v, \ell_v=2}x^-_{v})+\gamma(\sum_{v, \ell_v=1}(1-x^+_{v})+\sum_{v, \ell_v=2}x^+_{v}) = r^*.
\end{equation}
Note that since $r_{j-1} = \sum_{v, \ell_v=1}(1-x^-_{v})+\sum_{v, \ell_v=2}x^-_{v}$, $r_{j}=\sum_{v, \ell_v=1}(1-x^+_{v})+\sum_{v, \ell_v=2}x^+_{v}$ and $r_{j} \le r < r_{j-1}$, such $\gamma$ exists. Let $(x^*,y^*)$ be the solution obtained by taking a linear combination of $(x^-,y^-)$ and $(x^+,y^+)$ with weight $\gamma$. More formally,
\begin{align}
(x^*,y^*) &= (1-\gamma)(x^-,y^-)+\gamma(x^+,y^+) \label{eq:gamma2-1}\\
r^* &= \sum_{v, \ell_v=1}(1-x_{v}^*)+\sum_{v, \ell_v=2}x^*_{v}\label{eq:gamma2-2}
\end{align}
Since both $(x^-,y^-)$ and $(x^+,y^+)$ are optimal solutions of LP~\ref{tab:lagrangian}, it follows that $(x^*,y^*)$ is also an optimal solution to LP~\ref{tab:lagrangian}. We show that $(x^*,y^*)$ is an optimal solution to LP~\ref{tab:maxmove2-lp}, as well. To see this, suppose that $(x',y')$ is an optimal solution to LP~\ref{tab:maxmove2-lp}, while $(x^*,y^*)$ is not an optimal solution to LP~\ref{tab:maxmove2-lp}. Then, we have $\sum_{v, \ell_v=1}(1-x_{v}')+\sum_{v, \ell_v=2}x_{v}'\le r^*$ and $\sum_{e\in E(G)} c_ey_e'< \sum_{e\in E(G)} c_ey^*_e$. Combining this with equation~\eqref{eq:gamma2-2} we have
\begin{align*}
\alpha'(\sum_{v, \ell_v=1}(1-x_{v}')+\sum_{v, \ell_v=2}x_{v}') + \sum_{e\in E(G)} c_ey_e'<\alpha'(\sum_{v, \ell_v=1}(1-x_{v}^*)+\sum_{v, \ell_v=2}x^*_{v})+\sum_{e\in E(G)} c_ey^*_e, 
\end{align*}
implying that $(x',y')$ is a solution to LP~\ref{tab:lagrangian} with a better objective value than $(x^*,y^*)$, which is a contradiction. Thus, $(x^*,y^*)$ is an optimal solution to LP~\ref{tab:maxmove2-lp} and we have that 
\begin{equation}\label{eq:helper3}
    \sum_{e\in E(G)}c_ey_e^*\le \delta(S^*).
\end{equation}
By equation \eqref{eq:gamma}, we have that $|S_{j-1}|=\sum_{v, \ell_v=1}(1-x_{v}^-)+\sum_{v, \ell_v=2}x_{v}^-\le r^*/(1-\gamma)$. By the lemma's assumption, we have $|S_{j-1}|\ge r^*/(1-\lambda)$; therefore, it holds that $\gamma\ge \lambda$. Using equation \eqref{eq:gamma2-1}, we have that $y_e^*\ge \gamma y^+_e\ge \lambda y_e^+$, and so 
$$
\delta(S_{j})=\sum_{e\in E(G)} c_ey_e^+\le \frac{1}{\lambda}\sum_{e\in E(G)} c_ey_e^*\le \frac{1}{\lambda}\delta(S^*),
$$
where the last inequality uses inequality \eqref{eq:helper3}.
\end{proof}

To finish the proof of Theorem~\ref{thm:2part}, let $\lambda = 1-\frac{r^*}{r+1}$. Then $|S_{j-1}|\ge \frac{1}{1-\lambda}r^*= r+1$, which holds by the construction of Algorithm~\ref{alg:2part}. Therefore, by Lemma \ref{lem:helper} we have $\delta(S_{j})\le \frac{1}{\lambda}\delta(S^*)$.
Note that $\frac{1}{\lambda}=\frac{r+1}{r+1-r^*}$, concluding the proof of the theorem.
\section{HARDNESS RESULTS}\label{sec:hardness}
In this section, we state our hardness results, including lower bounds on the integrality gap of LP~\ref{tab:maxmove-lp} and our W[1]-hardness findings.
\subsection{Integrality Gap of the  $r$-Move $k$-Partitioning Linear Program}\label{sec:integrality-gap}



The following lemma discusses the integrality gap of the \maxmove~linear program (i.e., LP~\ref{tab:maxmove-lp}).
\begin{lemma}
    \label{lem:integrality_gap}
    LP \ref{tab:maxmove-lp} has an integrality gap of at least $r+1$.
\end{lemma}
\begin{proof}
Let $P_1: (v_1,\ldots,v_{r+2})$ be a path in graph $G=(V,E)$ such that node $v_1$ is in partition $1$ and the rest of the nodes on $P_1$ belong in partition $2$, see Figure \ref{fig:max-ip}. 
Let the weight of the edge $(v_1,v_2)$ be $\epsilon$ and the weight of all of the other edges on $P_1$ be one. Let $P_2: (v_{r+3},\ldots,v_n)$ be a path 
with edges all of weight one. All nodes of $P_2$ belong in partition $2$. Let $v_1$ and $v_n$ be the terminals for partition $1$ and $2$, respectively. The defined partitions introduce a cut that has a weight of $\epsilon$. Moving any node from $P_2$ to partition $1$ only increases the cut value, so any optimal integer solution keeps all the nodes in $P_2$ in their initial partition. Moving any proper subset of nodes of $P_1$ to partition $1$ increases the cut value to at least one. Furthermore, one cannot move all of the nodes in $P_1$ to partition $1$ without violating the move constraint. So an optimal integer solution to LP \ref{tab:maxmove-lp} does not move any node in this graph and has a cut value of $\epsilon$.

On the other hand, any (fractional) optimal solution of LP \ref{tab:maxmove-lp} on this graph has a cut value of at most $\frac{\epsilon}{r+1}$: For any node $v_i$, $i\in\{2,\ldots,r+2\}$, this is achieved by setting $X_{v_i}^1=\frac{r}{r+1}$ and $X_{v_i}^2=\frac{1}{r+1}$. Let the nodes in $P_2$ stay in partition 2, i.e., $X_{v_i}^1 = 0, X_{v_i}^2 = 1$ for $i\in \{r+3,\ldots,n\}$. The value of the cut in this solution is $\frac{\epsilon}{r+1}$. So the integrality gap is at least $r+1$. Note that this proof can be generalized to any $k>2$ by adding dummy partitions with singletons in them.
\end{proof}

\begin{figure}
    \centering
    \includegraphics{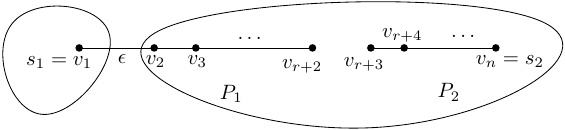}
    \caption{The example graph for the proof of Lemma~\ref{lem:integrality_gap}.}
    \label{fig:max-ip}
\end{figure}
Since the construction of Lemma \ref{lem:integrality_gap} is in fact a \stcutatmost construction, we have the following corollary.
\begin{corollary}
The linear program of the \stcutatmost~problem stated in \cite{zhang2016new} has an integrality gap of at least $r+1$.
\end{corollary}




\subsection{The $r$-Move $k$-Partitioning Problem is $W[1]$-hard}
We begin by introducing two new notations here that will be useful for the proofs of this section and then discuss the computational complexity of the \maxmove~problem. For a graph $G$ and any subsets of its nodes $C_1, C_2\subseteq V(G)$, let $E(C_1,C_2)$ refer to the set of edges between subsets $C_1$ and $C_2$.  For any node $v\in C_1$, we define $\mathrm{deg}_{C_1}(v)$ as the number of neighbors of node $v$ that are in $C_1$; that is, the number of nodes in $C_1$ that share an edge with node $v$.

\begin{restatable}{theorem}{reductionone}\label{thm:reduction-maxmove}
Given an $n$-node graph $G$ as an instance of the densest $r$-subgraph problem, there is an $O(n^2)$-node graph $G'$ with initial partitions $\{A,B\}$ 
such that the densest $r$-subgraph of $G$ has $m^*$ edges if and only if the optimal solution of the \maxmove[2]~problem on $G'$ reduces the initial cut value by $2m^*$. 
\end{restatable}
\begin{proof}
We reduce the densest $r$-subgraph problem to the \maxmove[2]~problem. Let $G$ be an instance of the densest $r$-subgraph problem that has $n$ nodes, $m$ edges and unit edge weights. Using $G$, we construct a graph $G'$ as an instance of the \maxmove[2]~problem. Graph $G'$ consists of two subgraphs $A$ and $B$, where $A$ is an exact copy of graph $G$ (with $n$ nodes and $m$ edges) and $B$ is a clique of size $2n^2+n$. Furthermore, $m$ nodes in $B$ are reserved for the $m$ edges of $A$; in other words, for each edge $e\in E(A)$ there exists a node $e\in V(B)$. Then, the endpoints of each edge $e=(u,v)\in E(A)$ gets connected to node $e\in V(B)$. Next, we add a terminal node $t$ to subgraph $A$ and connect it to all nodes in $A$. Similarly, a terminal node $s$ is added to subgraph $B$ and it gets connected to all nodes in $V(G')$. Note that after the addition of node $s$, subgraph $B$ becomes a $2n^2+n+1$ clique. This finishes the construction of graph $G'$, see Figure \ref{fig:reduction}. Let the \emph{initial cut} induced on $G'$ split the graph into two partitions $A$ and $B$. This cut has a value $2m+n< n^2+1$.

\begin{figure}
    \centering
    \includegraphics{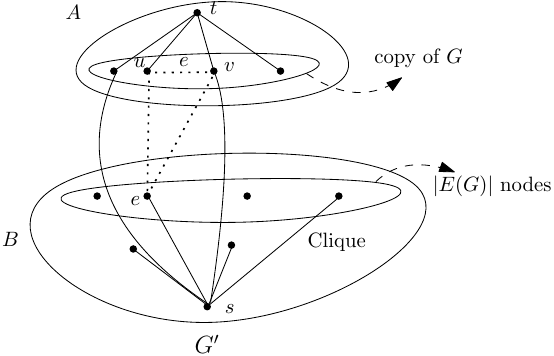}
    \caption{Reduction graph $G'$ created from a densest $r$-subgraph instance $G$. The solid-line edges do not depend on $G$. For each edge $e=(u,v)$ in $G$, nodes $u$ and $v$ in $A$ are connected to a node $e$ in $B$.}
    \label{fig:reduction}
\end{figure}

First, we show that if we move a subset of nodes $X\subseteq V(A)\setminus \{t\}$ (along with all edges incident to the nodes in $X$) to subgraph $B$, then the cut value reduces by $2\big|E(X)\big|$. Note that, if we move a set of nodes $X\subseteq V(A)\setminus \{t\}$
to subgraph $B$, then the cut value changes by $\big|E\big(X,V(A)\setminus X\big)\big|-\big|E\big(X,V(B)\big)\big|$. Since subgraph induced by $V(A)\setminus \{t\}$ is exactly the same as graph $G$, we have
\begin{align*}
   \big|E\big(X,V(A)\setminus X\big)\big| &= \big|E\big(X,\{t\}\big)\big|+\big|E\big(X,V(A)\setminus (X\cup \{t\})\big)\big| \\
   &= \big|X\big| + \sum_{v\in X} \big[\mathrm{deg}_{G}(v)-\mathrm{deg}_X(v)\big].
\end{align*}
Next, we have
\begin{align*}
    \big|E\big(X,V(B)\big)\big| &= \big|E\big(X,\{s\}\big)\big|+\big|E\big(X,V(B)\setminus\{s\}\big)\big| \\
    &= \big|X\big|+\sum_{v\in X} \mathrm{deg}_G(v);
\end{align*}
thus, $\big|E\big(X,V(A)\setminus X\big)\big|-\big|E\big(X,V(B)\big)\big| = -\sum_{v\in X} \mathrm{deg}_X(v)=-2\big|E(X)\big|$.

%

Now, let $X^*$ be the densest $r$-subgraph of $G$. If we move the copy of $X^*$ in $A$ from $A$ to $B$, then the cut value reduces by $2m^*=2\big|E(X^*)\big|$, as shown above. 
To prove the other direction, we show that if there exists a set of nodes $X\subseteq V(G')$ such that $|X|\le r$ and moving each node in $X$ to a partition different from the initial one it was assigned to reduces the cut by at least $2m^*$, then there exists an $r$-subgraph in $G$ with $m^*$ edges. To see this, first note that $X$ cannot have any nodes from subgraph $B$. This is because the resulting cut would have a value at least $2n^2+1$ which is bigger than the value of the initial cut; consequently, $X\subseteq V(A)$. As shown earlier in this proof, by moving $X$ from $A$ to $B$ the cut value is reduced by $2\big|E(X)\big|$, so $\big|E(X)\big|\ge m^*$. Therefore, the equivalent set of $X$ in $G$ has at least $m^*$ edges, hence $m^*=\big|E(X)\big|$. Finally, the time it takes to build graph $G'$ is $O\big(\big|V(G')\big|+\big|E(G')\big|\big)\subseteq O(n^4)$.
\end{proof}


\wonehardness*

\begin{proof}
For the \maxmove[2] problem, this comes from Theorem \ref{thm:reduction-maxmove} and the W[1]-hardness of the densest $r$-subgraph problem \cite{densest}. For the \maxmove~problem with $k>2$, by making the following modifications to graph $G'$ we can make the same reduction as that of the proof of Theorem~\ref{thm:reduction-maxmove} work: adding $k-2$ dummy partitions to graph $G'$, each containing one terminal that is not connected to any other node. Then, if $Y$ is the set of nodes that are moved to one of these dummy partitions $P$, then by moving $Y$ to one of the two main partitions the cut value will not increase. This is because $Y$ has no edges to the terminal node in $P$. Therefore, all solutions consist of moving nodes to one of the two main partitions. 
\end{proof}

\paragraph{Hardness of the \stcutatmost~problem:} We now slightly change the above construction to fit to the \stcutatmost~problem.
\begin{theorem}
Given an $n$-node graph $G$ as an instance of the densest $r$-subgraph problem, there is an $O(n^2)$-node graph $G'$ and a value $c(G')$ such that the densest $r$-subgraph of $G$ has $m^*$ edges if and only if the optimal solution of the \stcutatmost~problem on $G'$ has a value of $c(G')-2m^*$.
\end{theorem}
\begin{proof}
\begin{figure}
    \centering
    \includegraphics{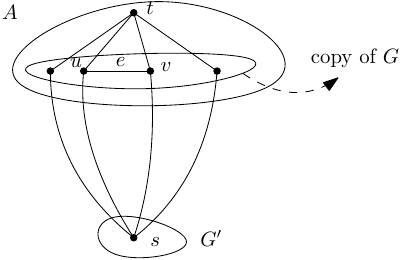}
    \caption{Reduction graph $G'$ created from a densest $r$-subgraph instance $G$.}
    \label{fig:reduction2}
\end{figure}
Let $A$ be a copy of the graph $G$. Add a terminal node $t$ to $A$ and connect it to all nodes in $A$ with unit-weight edges. Partition $B$ consists of only one node, terminal $s$. For each $v\neq t$ in $A$, connect node $v$ to terminal $s$ with an edge of weight $1+\mathrm{deg}_G(v)$, see Figure~\ref{fig:reduction2}. Let $c(G')$ be the cut value of partitions $\{A,B\}$ which is equal to $2\big|E(G)\big|+\big|V(G)\big|$.

Here, we show that moving any subset of nodes $X\in V(A)\setminus \{t\}$ from partition $A$ to $B$ reduces the value of the cut induced by partitions $\{A,B)\}$ by $2\big|E(X)\big|$. To see this, note that after moving nodes $X$ from $A$ to $B$, the cut value is reduced by $\big|E\big(X,V(A)\setminus X\big)\big|-\big|E\big(X,V(B)\big)\big|=\big(|X|+\sum_{v\in X}\mathrm{deg}_G(v)-\mathrm{deg}_X(v)\big) - \big(\sum_{v\in X} (\mathrm{deg}_G(v)+1)\big) = -2\big|E(X)\big|$.
Suppose $X^*$ is the densest $r$-subgraph in graph $G$ and $m^*=\big|E(X^*)\big|$. Moving $X^*$ to partition $B$ reduces the cut by $2m^*$. If the optimal solution to the \stcutatmost~problem on $G'$ is to move a set $X$ from $A$ to $B$, then this solution has a cut value of $c(G')-2\big|E(X)\big|$, and so we must have that $\big|E(X)\big|\ge m^*$. Since $|X|\le r$ and $X^*$ is the densest $r$-subgraph in $G$, we must have $\big|E(X)\big|=m^*$. Note that if $|X|<r$, then we can add arbitrary nodes to $X$ 
to make it have size $r$ nodes without reducing its number of edges. 
\end{proof}

From the above theorem and the W[1]-hardness of the densest $r$-subgraph problem we have that the weighted \stcutatmost~problem is W[1]-hard, thus, we have proved the following Corollary.

\begin{restatable}{corollary}{corwone}\label{cor:w1}
The \stcutatmost problem is W[1]-hard.
\end{restatable}

\paragraph{Hardness of the \maxmove~problem without terminals:} 

In this paper, we mostly consider the \maxmove~problem as a variant of the Multiway cut problem, i.e., we assume that the input graph has terminals that cannot be moved. However, one might wonder what the computational complexity of the \maxmove~problem without terminals is. All our algorithmic results carry over to the \maxmove~problem without terminals, since one can reduce the \maxmove~problem without terminals to the \maxmove~problem by adding dummy terminals for each partition as singletons. However, the main question in considering the \maxmove~problem without terminals is whether it can be solved faster than the \maxmove~problem, since in many partitioning problems the ``with terminal" version of the problem is harder than the ``without terminal" version. We show that our reduction works for the \maxmove~problem without terminals as well, so this problem is also W[1]-hard. Thus, one cannot expect the complexity to change drastically by removing the terminals.
\begin{restatable}{theorem}{notermhardness}\label{thm:noterm-hardness}
The \maxmove~problem without terminals is W[1]-hard.
\end{restatable}



\begin{proof}
Given a densest $r$-subgraph instance $G$, our reduction graph $G'$ is the same as that of Theorem \ref{thm:reduction-maxmove} except that we do not add terminal nodes $s$ and $t$ to the graph. The key observation here is that in the proof of Theorem \ref{thm:reduction-maxmove}  having terminal nodes $s$ and $t$ do not provide us with any specific benefit. We give a high level overview of the proof and the details can be easily  derived from the proof of Theorem \ref{thm:reduction-maxmove}. We argue that moving a set of nodes $X$ from partition $A$ to $B$ reduces the cut value by $2|E(X)|$. Then we can show that if the cut value is reduced by $2m^*$ after moving at most $r$ nodes, then these nodes must all be in $A$. This is because the nodes in $B$ create a clique and moving any of them to $A$ increases the cut value. Moreover, the set of nodes that are moved must induce a densest $r$-subgraph in $G$ in order for them to be an optimal solution to the $r$-move $k$-partitioning problem in $G'$.
\end{proof}



\section{Numerical Evaluation}\label{sec:exp}
We conduct a simple empirical assessment of our rounding algorithm (Algorithm~\ref{alg:r_approx}) of Theorem \ref{thm:m2approx} and FPTAS of Theorem \ref{thm:fptas} 
We remark that this experiment is \emph{not} an extensive empirical evaluation of the algorithms, and the focus of it is on the results that are not well explained by theory, namely that the approximation factor of Algorithm~\ref{alg:r_approx} in practice is better than what Theorem \ref{thm:m2approx} demonstrates. This suggests that there could be alternative analysis of our algorithm which results in better approximation guarantees for certain classes of graphs. 

\subsection{Experiments on Algorithm~\ref{alg:r_approx}}
%
\begin{figure}
    \centering
    \includegraphics[width=0.5\columnwidth]{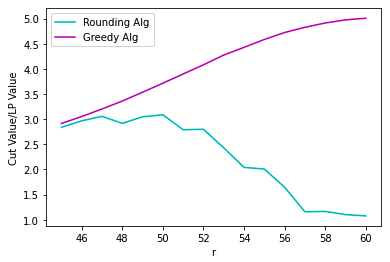}
    \caption{Performance of rounding (Algorithm~\ref{alg:r_approx}) and greedy algorithms with respect to the solution of LP \ref{tab:maxmove-lp} for the explained stochastic block graphs.}
    \label{fig:exp1}
\end{figure}
\paragraph{Set-up:} We generate our sample graphs using the stochastic block model on $90$ nodes as follows: First the $90$ nodes are divided into three equally sized clusters. Then, between any two nodes in the same cluster we add an edge with probability $p_H$. Similarly, between any two nodes in different clusters we put an edge with probability $p_L$. We let $p_H=0.3$ and $p_L= 0.1$. Note that this partitioning is very close to the optimal $k$-partitioning of the graph for $k=3$; hence, we re-partition the constructed graph into three partitions uniformly at random. These new random partitions are then set as the initial partitioning of the graph. Following these steps, we make $100$ random such graphs in total. As our benchmark, we consider the following simple greedy algorithm: The greedy algorithm has at most $r$ rounds. At each round, the algorithm moves the node that decreases the value of the $3$-cut by the largest amount. If at any point there is no such node, then the greedy algorithm halts. For each graph, we run Algorithm \ref{alg:r_approx} for $30$ different values of parameter $\rho$ and take the $3$-cut with the smallest value as the output.

\paragraph{Results:} For each of the 100 random graph and each of the two algorithms (greedy and Algorithm~\ref{alg:r_approx}), we compute the following ratios: The output of the algorithm divided by the objective value of LP \ref{tab:maxmove-lp}. For each value of $r$ in $[45,60]$, we compute the average value of this ratio over all graphs, see Figure \ref{fig:exp1}. For smaller values of $r$, we observe that both algorithms output similar cut values and they demonstrate similar performances. This could be due to the small size of our sample graphs and we believe the difference between the performances of these two algorithms is more evident with a larger sample size of graphs. While the greedy algorithm proves to perform reasonably well when the move parameter is bounded, we show that our FPTAS algorithm can beat greedy in this case. As $r$ approaches $60$ and above, the LP solution is integer in most instances; hence, the Algorithm~\ref{alg:r_approx} does not play a significant role\footnote{One can see that if $\hat{r}$ is the number of nodes needed to be moved in order to get the optimal $k$-cut solution (in the absence of any budgetary constraints), then $\mathbb{E}[\hat{r}]$ is near $60$.}. 


\subsection{Experiments for FPTAS}
\paragraph{Set-up:} The graph in this set of experiments is based on email data from a large European research institution~\cite{email-Eu-core_network}. The graph has 1005 nodes and each node belongs to one of the 42 departments at the research institute. There is an edge  between to nodes in the graph if those corresponding employees sent at least one email to each other. That is, emails are used as a proxy for communication between employees. There are 25,571 edges in the network. To downscale the size of the experiments, we take the induced sub-graph of the three largest departments. We take each department as a partition, thus the induced graph comes with a 3-partitioning. We then choose  5 to 10 nodes nodes of this graph, uniformly at random, and move them to a random partition. We perform this procedure 20 times, i.e., 20 different graphs are generated from the initial graph with 3-partitioning. On each of these graphs we run the greedy algorithm and the FPTAS to find the cut size resulted by each of these policies. We report the averaged the cut size for each policy as the cut size obtained from that method. Figure\ref{fig:email-data-FPTAS} shows the performance of the greedy algorithm and the FPTAS for two values of $\alpha$ with respect to the optimal solution of the problem which is computed by solving an Integer Program (which relaxes to LP \ref{tab:maxmove-lp}).

\begin{figure}
    \centering
    \includegraphics[width=0.6\columnwidth]{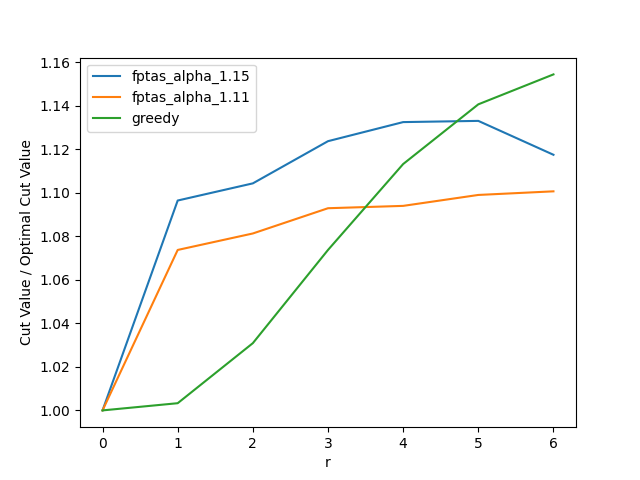}
    \caption{Performance of the FPTAS with two different values of $\alpha$ and greedy algorithm with respect to the solution of LP~\ref{tab:maxmove-lp} for the email-Eu-core network graph.}
    \label{fig:email-data-FPTAS}
\end{figure}

\paragraph{Results:} We observe that for small values of $r$ (that is, $r\le 3$) the greedy algorithm has better performance of either FPTAS's. However, as $r$ increases, the performance of the greedy algorithm deteriorates, while the performance of the FPTAS's stays steady or even improves. This is not surprising as with small values of $r$, the greedy algorithm has a limited number of moves to make and acting myopically most likely results in a good outcome. However, with larger values of $r$ it becomes more likely that a myopic initial move take the greedy algorithm farther from the optimal solution. Furthermore, as expected, we see that smaller $\alpha$ provides a better performance for the FPTAS, which comes at the expense of a higher running time.


\section{CONCLUSION}
This paper studies the \maxmove~problem. We show that this problem is $W[1]$-hard and give simple and practical approximation algorithms for it. These algorithms are: an FPTAS for constant $r$, a $3(r+1)$-approximation algorithm and an $(O(1),O(1))$ bicriteria algorithm for general $r$. Our main results focus on LP rounding techniques. There remains several interesting open problems to understand the complexity of the \maxmove~problem, some of them are listed below. 

(1) Is there an approximation algorithm for the \maxmove~problem whose running time and approximation factor are independent of $r$? Recall that there exists an $O(\log(n))$-approximation algorithm for the \stcutatmost~problem using tree decomposition techniques. Could one generalize this algorithm to the \maxmove~problem? 

(2) Can one find any approximation lower bound for this problem? Note that similar partitioning problems, such as the MinSBCC problem \cite{hayrapetyan2005unbalanced}, do not have any approximation lower bounds yet.





\bibliographystyle{unsrt}
\bibliography{bib} 

\begin{thebibliography}{10}

\bibitem{vision}
Sven Dickinson, Marcello Pelillo, and Ramin Zabih.
\newblock Introduction to the special section on graph algorithms in computer vision.
\newblock {\em IEEE Transactions on Pattern Analysis \& Machine Intelligence}, 23(10):1049--1052, 2001.

\bibitem{parallel1}
Bruce Hendrickson and Tamara~G Kolda.
\newblock Graph partitioning models for parallel computing.
\newblock {\em Parallel computing}, 26(12):1519--1534, 2000.

\bibitem{learning}
Thorsten Joachims.
\newblock Transductive learning via spectral graph partitioning.
\newblock In {\em Proceedings of the 20th international conference on machine learning (ICML-03)}, pages 290--297, 2003.

\bibitem{vlsi}
Andrew~B Kahng, Jens Lienig, Igor~L Markov, and Jin Hu.
\newblock {\em VLSI physical design: from graph partitioning to timing closure}.
\newblock Springer Science \& Business Media, 2011.

\bibitem{king2012geo}
Douglas~M King, Sheldon~H Jacobson, Edward~C Sewell, and Wendy K~Tam Cho.
\newblock Geo-graphs: an efficient model for enforcing contiguity and hole constraints in planar graph partitioning.
\newblock {\em Operations Research}, 60(5):1213--1228, 2012.

\bibitem{epidemiology}
Naman Shah, Matthew Malensek, Harshil Shah, Shrideep Pallickara, and Sangmi~Lee Pallickara.
\newblock Scalable network analytics for characterization of outbreak influence in voluminous epidemiology datasets.
\newblock {\em Concurrency and Computation: Practice and Experience}, 31(7):e4998, 2019.

\bibitem{dahlhaus1992complexity}
Elias Dahlhaus, David~S Johnson, Christos~H Papadimitriou, Paul~D Seymour, and Mihalis Yannakakis.
\newblock The complexity of multiway cuts.
\newblock In {\em Proceedings of the twenty-fourth annual ACM symposium on Theory of computing}, pages 241--251, 1992.

\bibitem{nphardmultiwaycut}
Elias Dahlhaus, David~S. Johnson, Christos~H. Papadimitriou, Paul~D. Seymour, and Mihalis Yannakakis.
\newblock The complexity of multiterminal cuts.
\newblock {\em SIAM Journal on Computing}, 23(4):864--894, 1994.

\bibitem{dinitz2022fair}
Michael Dinitz, Aravind Srinivasan, Leonidas Tsepenekas, and Anil Vullikanti.
\newblock Fair disaster containment via graph-cut problems.
\newblock In {\em International Conference on Artificial Intelligence and Statistics}, pages 6321--6333. PMLR, 2022.

\bibitem{bloch2023local}
Andrew Bloch-Hansen, Nasim Samei, and Roberto Solis-Oba.
\newblock A local search approximation algorithm for the multiway cut problem.
\newblock {\em Discrete Applied Mathematics}, 338:8--21, 2023.

\bibitem{LP_paper}
Gruia C{\u{a}}linescu, Howard Karloff, and Yuval Rabani.
\newblock An improved approximation algorithm for multiway cut.
\newblock In {\em Proceedings of the thirtieth annual ACM symposium on Theory of computing}, pages 48--52, 1998.

\bibitem{angelidakis2017improved}
Haris Angelidakis, Yury Makarychev, and Pasin Manurangsi.
\newblock An improved integrality gap for the c{\u{a}}linescu-karloff-rabani relaxation for multiway cut.
\newblock In {\em International Conference on Integer Programming and Combinatorial Optimization}, pages 39--50. Springer, 2017.

\bibitem{vazirani}
Huzur Saran and Vijay~V Vazirani.
\newblock Finding k cuts within twice the optimal.
\newblock {\em SIAM Journal on Computing}, 24(1):101--108, 1995.

\bibitem{ugchardness}
Rajsekar Manokaran, Joseph Naor, Prasad Raghavendra, and Roy Schwartz.
\newblock Sdp gaps and ugc hardness for multiway cut, 0-extension, and metric labeling.
\newblock In {\em Proceedings of the fortieth annual ACM symposium on Theory of computing}, pages 11--20, 2008.

\bibitem{bestapprox}
Ankit Sharma and Jan Vondr{\'a}k.
\newblock Multiway cut, pairwise realizable distributions, and descending thresholds.
\newblock In {\em Proceedings of the forty-sixth annual ACM symposium on Theory of computing}, pages 724--733, 2014.

\bibitem{ugc}
Subhash Khot.
\newblock On the power of unique 2-prover 1-round games.
\newblock In {\em Proceedings of the thiry-fourth annual ACM symposium on Theory of computing}, pages 767--775, 2002.

\bibitem{berczi2020improving}
Krist{\'o}f B{\'e}rczi, Karthekeyan Chandrasekaran, Tam{\'a}s Kir{\'a}ly, and Vivek Madan.
\newblock Improving the integrality gap for multiway cut.
\newblock {\em Mathematical Programming}, 183(1):171--193, 2020.

\bibitem{zhang2016new}
Peng Zhang.
\newblock A new approximation algorithm for the unbalanced min s--t cut problem.
\newblock {\em Theoretical Computer Science}, 609:658--665, 2016.

\bibitem{chen2016size}
Wenbin Chen, Nagiza~F Samatova, Matthias~F Stallmann, William Hendrix, and Weiqin Ying.
\newblock On size-constrained minimum s--t cut problems and size-constrained dense subgraph problems.
\newblock {\em Theoretical Computer Science}, 609:434--442, 2016.

\bibitem{lokshtanov2013clustering}
Daniel Lokshtanov and D{\'a}niel Marx.
\newblock Clustering with local restrictions.
\newblock {\em Information and Computation}, 222:278--292, 2013.

\bibitem{densest}
Leizhen Cai.
\newblock Parameterized complexity of cardinality constrained optimization problems.
\newblock {\em The Computer Journal}, 51(1):102--121, 2008.

\bibitem{goldschmidt1994polynomial}
Olivier Goldschmidt and Dorit~S Hochbaum.
\newblock A polynomial algorithm for the k-cut problem for fixed k.
\newblock {\em Mathematics of operations research}, 19(1):24--37, 1994.

\bibitem{armon2006multicriteria}
Amitai Armon and Uri Zwick.
\newblock Multicriteria global minimum cuts.
\newblock {\em Algorithmica}, 46(1):15--26, 2006.

\bibitem{watanabe1987edge}
Toshimasa Watanabe and Akira Nakamura.
\newblock Edge-connectivity augmentation problems.
\newblock {\em Journal of Computer and System Sciences}, 35(1):96--144, 1987.

\bibitem{karger1996new}
David~R Karger and Clifford Stein.
\newblock A new approach to the minimum cut problem.
\newblock {\em Journal of the ACM (JACM)}, 43(4):601--640, 1996.

\bibitem{fomin2013parameterized}
Fedor~V Fomin, Petr~A Golovach, and Janne~H Korhonen.
\newblock On the parameterized complexity of cutting a few vertices from a graph.
\newblock In {\em International Symposium on Mathematical Foundations of Computer Science}, pages 421--432. Springer, 2013.

\bibitem{zhang2014unbalanced}
Peng Zhang.
\newblock Unbalanced graph cuts with minimum capacity.
\newblock {\em Frontiers of Computer Science}, 8:676--683, 2014.

\bibitem{racke2008optimal}
Harald R{\"a}cke.
\newblock Optimal hierarchical decompositions for congestion minimization in networks.
\newblock In {\em Proceedings of the fortieth annual ACM symposium on Theory of computing}, pages 255--264, 2008.

\bibitem{feige2003cutting}
Uriel Feige, Robert Krauthgamer, and Kobbi Nissim.
\newblock On cutting a few vertices from a graph.
\newblock {\em Discrete Applied Mathematics}, 127(3):643--649, 2003.

\bibitem{bonnet2015multi}
Edouard Bonnet, Bruno Escoffier, Vangelis~Th Paschos, and Emeric Tourniaire.
\newblock Multi-parameter analysis for local graph partitioning problems: Using greediness for parameterization.
\newblock {\em Algorithmica}, 71(3):566--580, 2015.

\bibitem{feldmann2015balanced}
Andreas~Emil Feldmann and Luca Foschini.
\newblock Balanced partitions of trees and applications.
\newblock {\em Algorithmica}, 71(2):354--376, 2015.

\bibitem{chuzhoy2015approximation}
Julia Chuzhoy, Yury Makarychev, Aravindan Vijayaraghavan, and Yuan Zhou.
\newblock Approximation algorithms and hardness of the k-route cut problem.
\newblock {\em ACM Transactions on Algorithms (TALG)}, 12(1):1--40, 2015.

\bibitem{hayrapetyan2005unbalanced}
Ara Hayrapetyan, David Kempe, Martin P{\'a}l, and Zoya Svitkina.
\newblock Unbalanced graph cuts.
\newblock In {\em European Symposium on Algorithms}, pages 191--202. Springer, 2005.

\bibitem{gallo1989fast}
Giorgio Gallo, Michael~D Grigoriadis, and Robert~E Tarjan.
\newblock A fast parametric maximum flow algorithm and applications.
\newblock {\em SIAM Journal on Computing}, 18(1):30--55, 1989.

\bibitem{wong1983combinatorial}
Richard~T Wong.
\newblock Combinatorial optimization: Algorithms and complexity (christos h. papadimitriou and kenneth steiglitz).
\newblock {\em SIAM Review}, 25(3):424, 1983.

\bibitem{email-Eu-core_network}
Jue Leskovec.
\newblock email-\text{E}u-core network dataset.

\end{thebibliography}
\end{document}